\documentclass{article}
\usepackage{graphicx} 
\usepackage{bm}
\usepackage{tikz, pgfplots}
    \pgfplotsset{compat=1.3}
\usepackage{bbold, mathtools, amsthm,tikz}
\usepackage[]{url}
\usepackage[colorlinks,hypertexnames=true,citecolor=blue]{hyperref} 
\hypersetup{breaklinks=true}
\usepackage{float}
\usepackage{array}
\usepackage{pgfplots}
\pgfplotsset{width=5.5cm,compat=1.9}
\usepackage{natbib}
\usepackage[title,titletoc]{appendix}
\usepackage{multicol}
\usepackage{todonotes}
\setcitestyle{authoryear,open={(},close={)}}
\newtheorem{theorem}{Theorem}
\usepackage{adjustbox} 
\usepackage{subcaption}
\usepackage{amssymb}
\usepackage{tikz}
\usetikzlibrary{shapes, arrows, positioning, fit, calc, shapes.geometric}
\usepackage{amsmath}\usepackage{graphicx}
\usepackage{booktabs}

\newtheorem{lemma}{Lemma}
\newtheorem{remark}{Remark}

\newtheorem{observation}{Observation}
\newtheorem{corollary}{Corollary}
\usepackage{colortbl} 
\usepackage{wrapfig}     
\usepackage{booktabs}    
\usepackage{lipsum}   
\usepackage{microtype}        
\usepackage{array} 
\usepackage[table]{xcolor}
\newenvironment{keyword}
{\par\noindent\textbf{Keywords:}}
{\par}

\title{Driver Identification and PCA-Augmented Selection–Shrinkage Framework for Nordic System Price Forecasting}

\author{Yousef Adeli Sadabad, Mohammad Reza Hesamzadeh,\\ György Dán, Matin Bagherpour, Darryl R. Biggar\footnote{Y. Adeli Sadabad (yoas@kth.se), M. R. Hesamzadeh (mrhesa@kth.se), G. Dán (gyuri@kth.se) are with KTH Royal Institute of Technology (Sweden), M. Bagherpour (Matin.Bagherpour@nordpoolgroup.com) is with Oslo University and Nord Pool (Norway), and D. R. Biggar (Darryl.Biggar@monash.edu) is with Monash University (Australia).}}
\date{}

\begin{document}

\maketitle

\begin{abstract}
The System Price (SP) of the Nordic electricity market serves as a key reference for financial hedge contracts such as Electricity Price Area Differentials (EPADs) and other risk management instruments. Therefore, the identification of drivers and the accurate forecasting of SP are essential for market participants to design effective hedging strategies. This paper develops a \textit{systematic framework} that combines interpretable drivers analysis with robust forecasting methods. It proposes an interpretable feature engineering algorithm to identify the main drivers of the Nordic SP based on a novel combination of K-means clustering, Multiple Seasonal-Trend Decomposition (MSTD), and Seasonal Autoregressive Integrated Moving Average (SARIMA) model. Then, it applies principal component analysis (PCA) to the identified data matrix, which is adapted to
the downstream task of price forecasting to mitigate the issue of imperfect multicollinearity in the data. Finally, we propose a multi-forecast selection-shrinkage algorithm for Nordic SP forecasting, which selects a subset of complementary forecast models based on their bias-variance tradeoff \textit{at the ensemble level} and then computes the optimal weights for the retained forecast models to minimize the error variance of the combined forecast. Using historical data from the Nordic electricity market, we demonstrate that the proposed approach outperforms individual input models uniformly, robustly, and significantly, while maintaining a comparable computational cost. Notably, our systematic framework produces superior results using simple input models, outperforming the state-of-the-art Temporal Fusion Transformer (TFT). Furthermore, we show that our approach also exceeds the performance of several well-established practical forecast combination methods. 
\end{abstract}


\section{Introduction}
\label{Introduction}
Electricity prices are highly volatile and exhibit complex underlying dynamics \cite{bunn2004modelling}. These dynamics are driven by factors such as supply–demand imbalances, weather-dependent generation, fuel-price fluctuations~\cite{afanasyev2021fundamental}, and uncertainties in physical infrastructure~\citep{mosquera2019drivers, maciejowska2020assessing}. As a result, even day-ahead price forecasting remains a persistent challenge for system operators and market participants~\cite{karakatsani2008forecasting}. At the same time, understanding the main driving factors and developing improved forecasting algorithms can significantly enhance decision-making by increasing the predictability of future electricity prices~\cite{EI2006,EI2016,busch2023development}.

Electricity-price forecasting algorithms can be broadly categorized into statistical regression models and machine learning approaches \cite{weron2014electricity}. Despite numerous methods proposed in the literature, several challenges exist, particularly regarding (1) feature selection, (2) multicollinearity, and (3) the design of theoretically justified and robust ensemble models. These key challenges have been explored to varying degrees in the existing literature. 

Authors in \cite{lago2021forecasting} provide an in-depth review of electricity price forecasting techniques. Notably, some reviewed research suffers from misspecification problems, using concurrent relationships and reporting surprisingly low errors. Authors in~\cite{raviv2015forecasting} and~\cite{maciejowska2015forecasting} demonstrate the usefulness of using intraday hourly prices for forecasting the average daily day-ahead prices in the Nordic and the PJM markets. Their analyses, however, do not extend to forecasting across different delivery periods. Reference~\cite{marcjasz2018selection} compares the performance of models with univariate and multivariate structures, without features such as load and production categories, and shows a minor edge in the multivariate framework. However, they show that it does not uniformly outperform the univariate one. Our correlation analysis in Section \ref{sec:correlation} shows non-trivial dependencies across delivery periods, and it is crucial to use past values of other delivery periods to predict the price of a specific delivery period. This motivates the use of a panel data approach for system price forecasting.

In recent years, researchers have demonstrated the superiority of regression models with a large number of input features that utilize regularization techniques as implicit feature selection methods~\cite{ziel2018day}, and~\cite{ ziel2016forecasting}. The authors in~\cite{uniejewski2016automated} show that lasso and elastic net regressions outperform standard regression models by implicitly selecting features. \cite{MIRAKYAN2017228} acknowledge the issue of multicollinearity in electricity price forecasting and apply ridge regression to mitigate its effects. However, the paper does not provide a detailed analysis or explanation of the sources and mechanisms through which multicollinearity arises in electricity market data. In the current article, we show that the fundamental reason for using regularization techniques is the presence of imperfect multicollinearity in the panel data approach to electricity price forecasting. We also show that feature selection
and regularization can be separated, with an emphasis on interpretable feature selection and integrating regularization with the downstream task.

The authors in \cite{pourdaryaei2024new} use a hybrid feature selection method that combines Mutual Information (MI) and a Neural Network (NN). The MI is first applied to filter relevant, non-redundant input variables. A NN then selects the final optimal subset of features from this filtered set. While MI is a technique that can quantify the relationship between variables, the use of a neural network for the final selection is a nonlinear process, which can make the final feature choices less directly interpretable. Therefore, the overall feature selection process is not interpretable. 

Another important direction in the price-forecasting algorithms is to combine individual forecasting models. References \cite{hubicka2018note} and \cite{marcjasz2018selection} propose forecast combinations based on short-long calibration windows. Authors in \cite{nowotarski2014empirical} did a comprehensive empirical study on forecast averaging, demonstrating the usefulness of forecast averaging methods. They found that equal-weighting is a simple yet effective method when no single predictor dominates, while the Constrained Least Squares (CLS) method offers a good balance between robustness and accuracy. In contrast, methods such as Ordinary Least Squares (OLS) and Bayesian Model Averaging (BMA) were shown to be unsuitable for day-ahead price forecasting. However, there is no clear conclusion on which methods are superior, how to develop proper sub-models, or how to weight them. More importantly, the final model should be conclusive and perform uniformly across different delivery periods, which is not the case in \cite{nowotarski2014empirical}. In our paper, we suggest an approach that produces robust results in all delivery periods; none of the reviewed papers report that the final model outperforms the input models or the baseline across different delivery periods uniformly and robustly. Similarly, \citet{MIRAKYAN2017228} propose a composite approach combining ridge regression, neural networks, and support vector regression with weighting schemes such as inverse RMSE (IRMSE) and CLS to improve seasonal robustness. However, their model selection and combination strategies lack theoretical justification, leaving open the question of how to design ensembles that are both interpretable and theoretically grounded, an issue we explicitly address in our framework.

\citet{marcos2020short} and \citet{nitka2021forecasting} suggest that electricity price time series exhibit recurrent regimes, using $k$-means clustering \cite{marcos2020short} and $k$-nearest neighbors clustering \cite{nitka2021forecasting} to identify these regimes and then calibrate models on data segments resembling current conditions. However, given the substantial changes in the generation mix over the past decade, as well as evolving market characteristics such as increased price volatility and the occurrence of negative prices, calibrating models solely on historical segments risks producing biased factor estimates. In contrast, our approach uses $k$-means clustering not to restrict model training to past segments but to identify short-term and long-term drivers. The models are then allowed to determine the statistical significance and relative importance of each driver autonomously, thereby avoiding bias and enhancing interpretability.

Motivated by these research gaps, in our paper, we suggest a systematic approach that produces robust results in all delivery periods. Since the prices for all delivery periods are jointly determined, we employ the panel data approach. We show that regularization at the model level addresses the imperfect multicollinearity in panel data forecasting, and feature selection and regularization should be separated, with an emphasis on interpretable feature selection and integrating regularization with the downstream forecasting task. The Nordic SP is a reference price for several financial hedge contracts, including Electricity Price Area Differential (EPAD) contracts. Hence, identifying its drivers and improving its forecasting accuracy are essential for market participants seeking to manage price risks.   

With this background, the main contributions of the current paper are as follows:

\begin{enumerate}
    \item It provides a detailed statistical exploration of the Nordic System Price (SP), presenting the Nordic SP stylized facts through several observations. The null hypothesis of the long-term stationarity of the Nordic SP is checked using the Augmented Dickey-Fuller (ADF) and Kwiatkowski-Phillips-Schmidt-Shin (KPSS) tests, and the stability of the seasonal pattern in the Nordic SP is examined through the Canova-Hansen (CH) test.
    
    \item An interpretable feature-engineering algorithm is proposed to find the drivers of the Nordic SP. This algorithm is developed based on the K-means clustering, MSTD method, and the SARIMA model.
    
    \item The multi-forecast selection-shrinkage algorithm is proposed to forecast the Nordic SP. In the selection phase, it identifies the complementary forecast models from an ensemble of models, based on Theorem \ref{thm1} and Corollary \ref{crl1}. In the shrinkage phase, the optimal weights for combining the selected forecast models are computed via Theorem \ref{thm2}.
    
    \item We address the imperfect multicollinearity problem using Principal Component Analysis (PCA), and we show that PCA indeed leads to improved forecasting results. Additionally, as an alternative approach to the commonly used elbow method, we propose to find the optimal number of components based on the downstream forecasting task. Together, these contributions establish a \textit{systematic approach} to driver identification and robust forecasting of the Nordic system prices as shown in Fig. \ref{fig:model_flowchart}. 
\end{enumerate}

Table~\ref{tab:comparison} compares our contributions with existing papers in the literature. In most cases, the aspects checked are only partially addressed in the papers. The table is structured around seven key aspects, which we review in turn below. 

For the multicollinearity case (first column), we consider the application of regularization methods in regressions at the model level. However, none of the reviewed papers clearly state or analyze why and how multicollinearity arises in the electricity market data. 

For interpretable feature selection (second column), there are a few papers on electricity price forecasting addressing this aspect in detail, and, as we mentioned above, they mostly rely on regularization as implicit feature selection, which is a problem for practitioners who are interested in understanding the dynamics. 

In the panel data column (third column), papers that consider 24 separate univariate models are included. Yet, as we mentioned above, the complex relationship between delivery periods necessitates a panel data approach, while the reviewed studies account only for specific interconnections through selected variables. 

In the fourth column, the separation of feature selection from regularization, as we mentioned, has not been previously addressed. In the literature, we see regularization and feature selection typically integrated; this reduces the effect of regularization and also undermines interpretability. We will show that, when carefully designed, interpretable feature selection significantly reduces the need for regularization.

In the case of uniform performance (fifth column), none of the reviewed papers reports consistent outperformance across all delivery periods relative to input models or baselines. 

Regarding model selection (sixth column), we note that none of the reviewed studies explicitly provides a theoretical framework for model development and selection in the context of ensemble forecasting. Instead, the works listed in this column typically focus on combining a diverse set of well-established models or describing heuristic procedures for model aggregation.

In the case of seasonal stationarity analysis (seventh column), this aspect is overlooked in all studied papers. We therefore included only those studies that use dummy variables for seasonality. In electricity markets, seasonality can occur daily, weekly, or yearly. In the day-ahead market, since the data frequency is daily, the most relevant seasonality component for short-term predictions is weekly seasonality. Consequently, statistical models should explicitly test for and account for stationarity at seasonal frequencies.

As Table \ref{tab:comparison} clearly highlights, while existing studies partially address some aspects, none provide a unified and systematic approach. Our paper contributes by addressing all seven aspects simultaneously.

\begin{table}[ht]
\centering
\footnotesize 
\resizebox{\textwidth}{!}{%
\begin{tabular}{p{5cm} >{\centering\arraybackslash}p{2cm} >{\centering\arraybackslash}p{2cm} >{\centering\arraybackslash}p{2cm} >{\centering\arraybackslash}p{3cm} >{\centering\arraybackslash}p{2cm} >{\centering\arraybackslash}p{2cm} >{\centering\arraybackslash}p{2.2cm}}
\toprule
\textbf{Reference} & \textbf{\parbox{2cm}{Addressing\\Multi-\\collinearity}} & \textbf{\parbox{2cm}{Interpretable\\Feature\\Selection}} & \textbf{\parbox{2cm}{Panel\\Data}} & \textbf{\parbox{2.5cm}{Separation\\ of Feature\\Selection\\ from Regularization}} & \textbf{\parbox{2cm}{Uniform\\Performance}} & \textbf{Model Selection} & \textbf{\parbox{2cm}{Seasonal\\(Stationarity) Analysis}} \\
\midrule
\cite{raviv2015forecasting} & & & & & $\checkmark$ & & $\checkmark$ \\
\cite{ziel2018day} & & & $\checkmark$ & & & & \\
\cite{marcjasz2018selection} & & & & & $\checkmark$ & $\checkmark$ & \\
\cite{hubicka2018note} & & & & & & & \\
\cite{lago2018forecasting} & & & & & &$\checkmark$ &$\checkmark$ \\
\cite{ziel2016forecasting} & $\checkmark$ & & & & $\checkmark$ & & \\
\cite{lehna2022forecasting} & & & & & & $\checkmark$ & $\checkmark$ \\
\cite{hong2012day} & & & & & $\checkmark$& & $\checkmark$ \\ 
\cite{olivares2023neural} & &$\checkmark$ & & &$\checkmark$ & & $\checkmark$\\
\cite{marcos2020short} & & & & & & & \\
\cite{pourdaryaei2024new} & & & & &$\checkmark$ & & $\checkmark$\\
\cite{uniejewski2016automated} & $\checkmark$ & & $\checkmark$ & & $\checkmark$ & & \\
\cite{MIRAKYAN2017228} & $\checkmark$ & & & & & $\checkmark$ & \\
\cite{kitsatoglou2024ensemble} & & & & &$\checkmark$ &$\checkmark$ &$\checkmark$ \\
\cite{jiang2023multivariable} & & & & & $\checkmark$&$\checkmark$ &$\checkmark$ \\
\cite{nowotarski2014empirical} & & & & & & & $\checkmark$ \\
\cite{nitka2021forecasting} & & & $\checkmark$ & & & & $\checkmark$ \\
\textbf{Our paper} & $\checkmark$ & $\checkmark$ & $\checkmark$ & $\checkmark$ & $\checkmark$ & $\checkmark$ & $\checkmark$ \\
\bottomrule
\end{tabular}%
}
\caption{Comparison of our contributions with existing papers in the literature: A checkmark ($\checkmark$) indicates that the aspect is addressed. }
\label{tab:comparison}
\end{table}

This paper is organized as follows: In Section \ref{sec:section2}, background and problem formulation are presented. Statistical characterization of the Nordic SP is discussed in Section \ref{sec:stat}. Section \ref{sec:Section IV} explains the forecast-optimized feature engineering approach. The multi-forecast selection-shrinkage algorithm is proposed in Section \ref{sec:multi-forecast}. Section \ref{sec:Numerical} provides a comprehensive analysis of the numerical results. Conclusions are presented in Section \ref{sec:conc}, while Appendix \ref{app:1} provides a concise background on the CH test. 

\section{Background and Problem Formulation}
\label{sec:section2}
The Nordic day-ahead electricity market is coupled with the European market through the Single Day-Ahead Coupling (SDAC) model. At 10:00 CET, the Transmission System Operator (TSO) publish the day ahead available transmission capacities (ATC), and market participants have until 12:00 CET to submit bids to Nord Pool. The Nordic SP is calculated on the basis of the results of the SDAC model, and all results are published at 12:45 CET. After validation of these results, the SP is confirmed around 13:00 CET.

Thus, the day-ahead hourly prices for delivery of electricity at day $d$ computed on day $d-1$ are best modeled as a 24-dimensional price vector $\textbf{p}_d=(p_{d,1}, p_{d,2},...,p_{d,24})$.
We denote the SP dataset up to day $d$ by $\textbf{P}_d = (\textbf{p}_1^T,...,\textbf{p}_d^T)^T$.
Importantly, information after 12:00 CET is not applicable for forecasting day-ahead SP. Let us denote the data available before 12:00 CET on the day $d-1$ concerning day $d$ by $\textbf{Y}_{d-1} = (\textbf{y}_1^T,...,\textbf{y}_{d-1}^T)^T$, e.g., daily transmission capacities, daily day-ahead consumption for the day $d-1$ which is calculated on day $d-2$, gas prices, etc., i.e., all daily data that can be used for forecasting $\textbf{p}_d$, with $\textbf{y}_d = (y_{d,1},...,y_{d,m})$, where $m$ is number of raw features. Accordingly, we have a significant-features vector $\textbf{x}$ with the representation $\textbf{x}_d = (x_{d,1},...,x_{d,n})$, where $n < m$ is the number of significant features and the matrix $\textbf{X}_d = (\textbf{x}_1^T,...,\textbf{x}_d^T)^T $ is the matrix of significant features. Essentially, significant features are raw features with significant explanatory power on the SP (identified in Section~\ref{sec:Section IV}). Finally, we have a reduced feature matrix $\textbf{Z}$ that is obtained after applying PCA to the matrix $[\textbf{P}|\textbf{X}]$.

Our objective is to solve the problem 
\begin{equation}
    \min_{\textbf{f},\phi}\sum_{d=L+1}^{D}(\textbf{p}_d-\textbf{f}(\phi(\textbf{Y}_{d-1}^L)))^2,
\end{equation}
where $L$ is the look back parameter and $\textbf{Y}_{d-1}^L$ is the sub-matrix of $\textbf{Y}_{d-1}$ consisting of the data of the last $L$ days, i.e., days $d-L$ to $d-1$. The output of the $\phi(\textbf{Y}_{d-1}^L)$ is either $\textbf{Z}_{d-1}^L=\phi(\textbf{Y}_{d-1}^L)$, which is the input for machine learning models, and we call it the reduced features data matrix with look back $L$, or $[\textbf{P}|\textbf{X}]$, which is used as input to statistical models, and $\textbf{f}$ is the forecasting model. We are interested in how to design $\phi$ and $\textbf{f}$.

\section{Statistical characterization of the Nordic SP} \label{sec:stat}

\subsection{Data set}
We use historical Nordic SP data from January 1, 2015, to May 31, 2024, covering 3,439 days with 24 hourly observations per day. Table~\ref{tab:DESCRIPTIVE-STATISTICS} summarizes the descriptive statistics, and Fig.~\ref{fig:SystemPrice_Plot} shows the historical SP series, which exhibits frequent price spikes and sustained periods of high volatility. Except for the first six delivery periods and the last one, the standard deviation of hourly prices exceeds the mean, indicating substantial dispersion and fat-tailed behavior—suggesting that volatility dynamics dominate mean dynamics in this market. Moreover, the mean values are close to the 75th percentile, implying that when deviations occur, they tend to be large. These stylized facts highlight the need for forecasting methods that can capture volatility-driven price dynamics rather than focusing solely on mean prediction.

\begin{figure}[t]
\centering
\includegraphics[width=0.7\columnwidth]{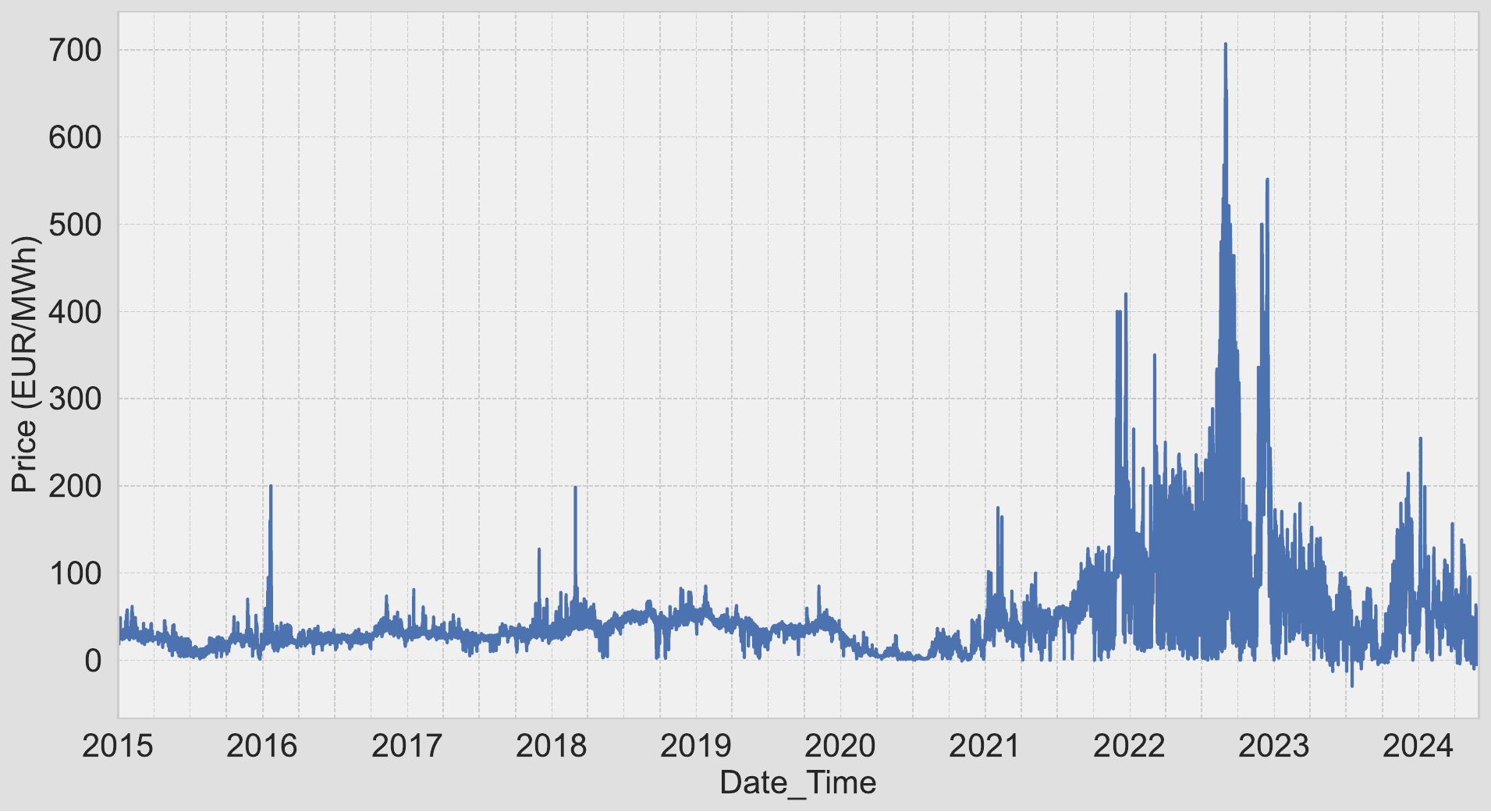}
\caption{Historical hourly values of the  Nordic system price.}
\label{fig:SystemPrice_Plot}
\end{figure}

\begin{table}[t]
\centering
\caption{\sc System Price Descriptive Statistics}
\label{tab:DESCRIPTIVE-STATISTICS}
\resizebox{\textwidth}{!}{%
\begin{tabular}{c c c c c c c c c c c c c c c c c c c c c c c c}
\toprule
& 0-1 & 1-2 & 2-3 & 3-4 & 4-5 & 5-6 & 6-7 & 7-8 & 8-9 & 9-10 & 10-11 & 11-12 \\
\hline
\midrule
mean & 37.16 & 35.22 & 33.97 & 33.38 & 34.10 & 36.96 & 42.30 & 48.40 & 51.83 & 51.05 & 49.37 & 47.45 \\
std & 33.61 & 31.27 & 29.90 & 29.30 & 30.23 & 33.62 & 43.02 & 52.35 & 57.70 & 56.68 & 53.93 & 50.43 \\
min & -4.38 & -5.49 & -6.06 & -5.71 & -4.74 & -3.07 & -2.58 & -1.19 & -2.78 & -9.14 & -11.74 & -13.03 \\
25\% & 22.73 & 21.32 & 20.44 & 20.06 & 20.26 & 22.09 & 24.51 & 26.59 & 27.73 & 27.77 & 27.35 & 26.77 \\
50\% & 29.80 & 28.85 & 28.22 & 27.98 & 28.44 & 29.91 & 32.13 & 34.99 & 36.45 & 36.06 & 35.58 & 35.03 \\
75\% & 39.85 & 38.47 & 37.71 & 37.49 & 37.94 & 40.08 & 44.01 & 49.65 & 52.04 & 51.23 & 49.60 & 47.80 \\
max & 381.98 & 339.84 & 326.18 & 284.31 & 283.29 & 304.96 & 561.49 & 688.32 & 700.00 & 679.93 & 651.68 & 560.48 \\
\hline
\midrule
& 12-13 & 13-14 & 14-15 & 15-16 & 16-17 & 17-18 & 18-19 & 19-20 & 20-21 & 21-22 & 22-23 & 23-0 \\
\hline
\midrule
mean & 45.57 & 44.11 & 43.56 & 44.30 & 46.39 & 50.14 & 51.33 & 50.63 & 48.45 & 46.25 & 43.38 & 38.97 \\
std & 47.31 & 46.03 & 45.75 & 47.18 & 50.18 & 55.57 & 57.10 & 57.19 & 54.31 & 49.71 & 43.78 & 36.08 \\
min & -19.96 & -25.10 & -29.90 & -25.04 & -11.06 & -2.74 & -1.69 & -1.80 & -1.91 & -2.06 & -2.57 & -3.31 \\
25\% & 26.07 & 25.52 & 25.00 & 25.12 & 25.65 & 26.77 & 27.34 & 27.12 & 26.50 & 26.02 & 25.14 & 23.79 \\
50\% & 34.33 & 33.73 & 33.19 & 33.38 & 34.13 & 35.64 & 36.19 & 35.48 & 34.40 & 33.59 & 32.33 & 30.61 \\
75\% & 46.41 & 44.98 & 44.80 & 45.38 & 47.49 & 50.47 & 50.96 & 49.93 & 47.50 & 45.76 & 43.77 & 40.58 \\
max & 534.26 & 520.18 & 533.10 & 549.91 & 563.69 & 660.60 & 689.84 & 706.87 & 677.79 & 648.01 & 542.26 & 375.79 \\
\bottomrule
\end{tabular}
}
\end{table}

The first occurrence of negative pricing was recorded on November 2, 2020, an event not observed in the price history dating back to 2012. The sharp price surge from late 2021 through 2023 can be partly attributed to the dramatic rise in European natural gas prices during this period. As shown in Section~\ref{sec:Section IV}, the Nordic region itself is not heavily reliant on natural gas for electricity production. Still, its price is influenced by cross-border coupling with countries where gas-fired generation is more significant. Therefore, the period from late 2021 to late 2023 is atypical for the Nordic SP and does not fully reflect the market’s underlying dynamics. When modeling the entire sample, regularization techniques are necessary to prevent overfitting this exceptional regime and ensure that forecasts generalize well.

Fig. \ref{fig:Daily} and Fig. \ref{fig:Yearly} show a box plot of the hourly Nordic SP per hour and per month, respectively, and show that there is considerable seasonality at both daily and yearly cycles. On an hourly basis, there is a morning peak between 8:00 am and 10:00 am and an evening peak between 5:00 pm and 7:00 pm. To remove daily seasonality, we subtract the average price of delivery period $h$, $\bar{p}_{.,h}=\frac{1}{D}\sum_{d=1}^{D}p_{d,h}$, from each hourly price $p_{d,h}$, and then from each resulting hourly price we subtract the average price of the corresponding month.

\begin{figure}[t!]  
    \centering
    
    \begin{minipage}{0.48\textwidth}
        \centering
        \includegraphics[width=\linewidth,height=5cm,keepaspectratio]{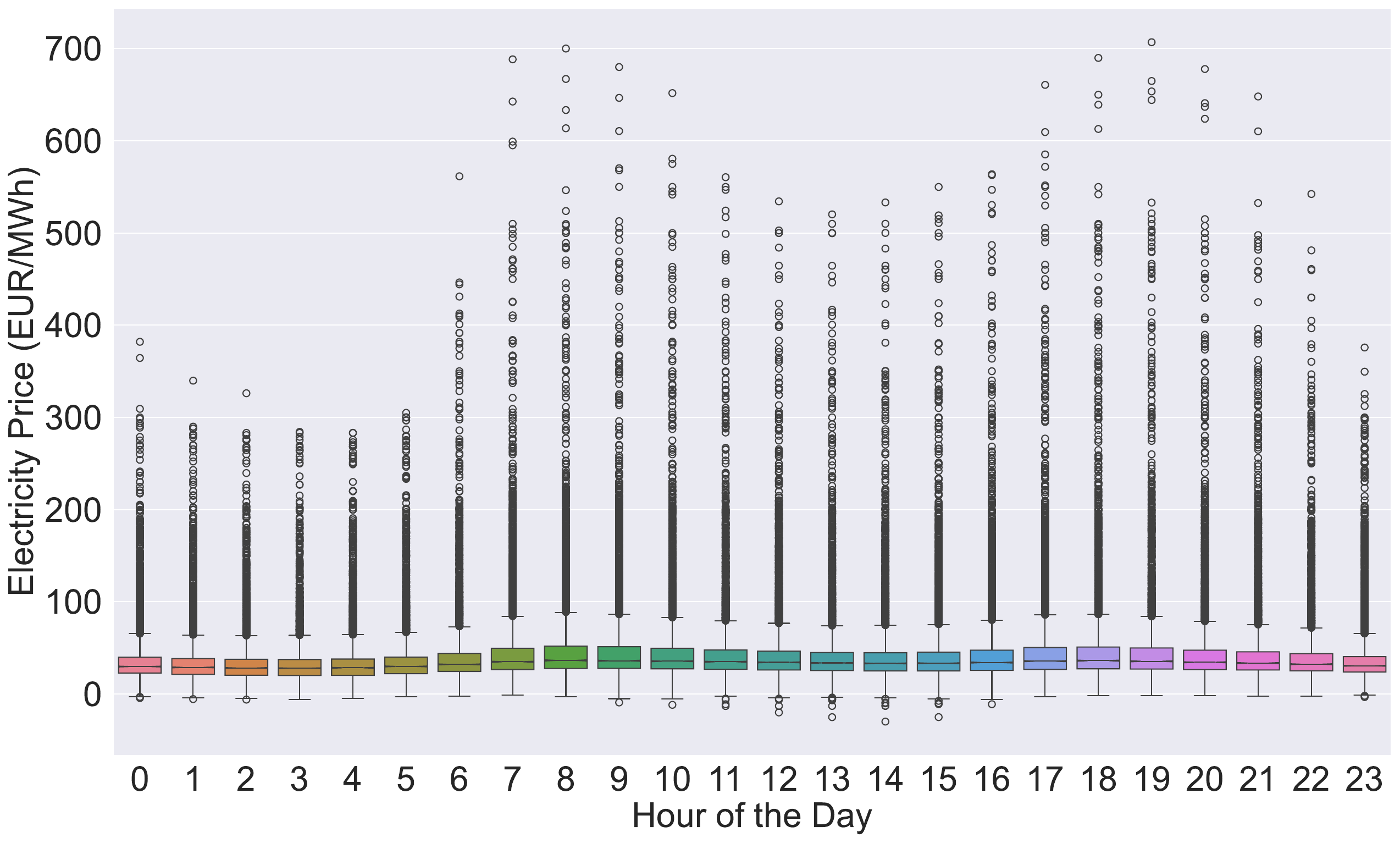}
        \subcaption{Daily seasonality}
        \label{fig:Daily}
    \end{minipage}%
    \hfill
    \begin{minipage}{0.48\textwidth}
        \centering
        \includegraphics[width=\linewidth,height=5cm,keepaspectratio]{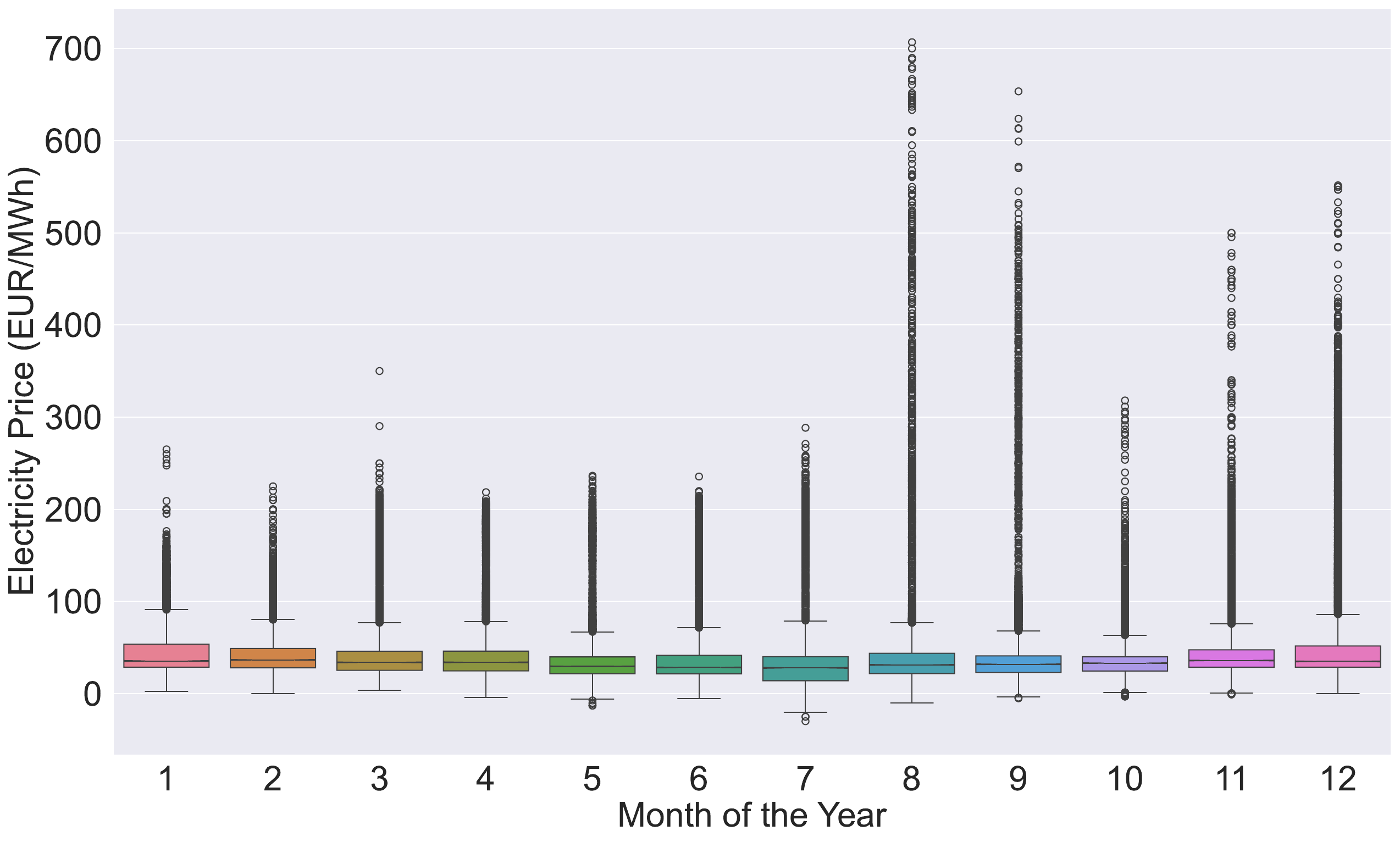}
        \subcaption{Yearly seasonality}
        \label{fig:Yearly}
    \end{minipage}
    
    \vspace{10pt}
    \begin{minipage}{\textwidth}
        \centering
        \includegraphics[width=0.8\linewidth,height=5cm,keepaspectratio]{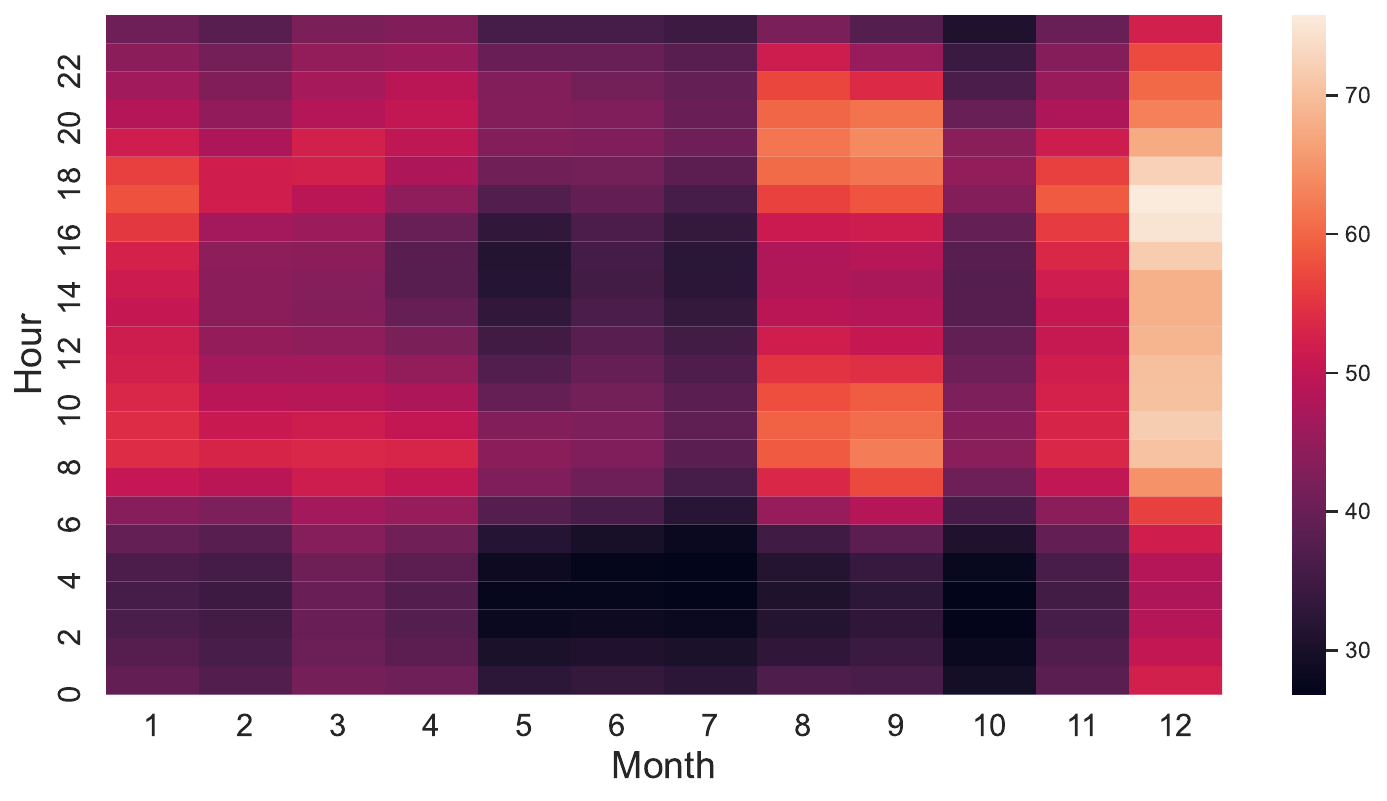}
        \subcaption{Heatmap of hour-Month}
        \label{fig:Heatmap}
    \end{minipage}
    
    \caption{Visualization of Nordic system price patterns: (a) shows daily seasonality, (b) displays yearly seasonality, and (c) presents a heatmap of monthly average hourly prices.}
    \label{fig:Seasonality}
\end{figure}

Fig. \ref{fig:Heatmap} shows a heat map of the monthly average hourly SP of each delivery period. The figure shows that the hourly SP exhibits an annual pattern. Specifically, prices across all delivery periods tend to be lower during May, June, and July. In contrast, prices in December are significantly higher than in other months, indicating significant seasonal variation.

We note that removing outliers should be performed cautiously in price analysis, as removing the highest or lowest 5\% of SP data would alter key statistical characteristics. In our work, we thus do not remove outliers.

\subsection{Correlation structure}
\label{sec:correlation}
Fig. \ref{fig:Correlation Structure} shows the correlation structure of the hourly prices over 24 hours. Fig. \ref{fig:Corr} shows the same day autocorrelation matrix, between prices of different hours within the same day. The figure shows a cluster that goes down after a jump at 6 am along the leading diagonal. Also, the two daily peaks are apparent from the off-diagonal clusters, and the lowest correlation is between the delivery periods 2-5 and the delivery periods 18-23. Fig. \ref{fig:Corr1} to Fig. \ref{fig:Corr7} show the autocorrelation matrices computed over prices with a lag of $i$ days. They show the correlation between $\textbf{p}_d$ (horizontal axis) and $\textbf{p}_{d-i}$ (vertical axis). From the lag-one autocorrelation matrix, we infer that, for example, for the first delivery period (hour zero to one), the lagged values of hours 12 to 23 have a more substantial impact on the price than itself. A similar pattern has been found before for the EPEX spot price for Germany and Austria, and the APX spot price for the Netherlands by \cite{ziel2016forecasting}. We also see the more complex pattern of the same kind in higher-order lags; for example, we can observe in Fig. \ref{fig:Corr7} that lagged values of the sixth delivery period have a higher correlation coefficient with the first four delivery periods than their own correlation. Therefore, in predicting the price of the specific delivery period of a day, we should use the lagged values of the prices of the other delivery periods. The analysis reveals that for the first six delivery periods of the day, the lagged values of some proceeding hours exhibit a higher correlation with the hourly price of that delivery period than the lagged values of the same delivery period. This finding suggests that, when forecasting hourly electricity prices, it is essential to consider the interconnections between the prices in different delivery periods rather than focusing solely on the lagged values of the target delivery period. Another important observation from the correlation structure of the Nordic system price is the coupling between hours 7 to 10 with hours 17 to 22; lagged values of one cluster are highly correlated with the current values of another cluster.

\begin{figure}[H]
    \centering
    \begin{subfigure}[b]{0.48\textwidth}
        \centering
        \includegraphics[width=\textwidth,height=4cm]{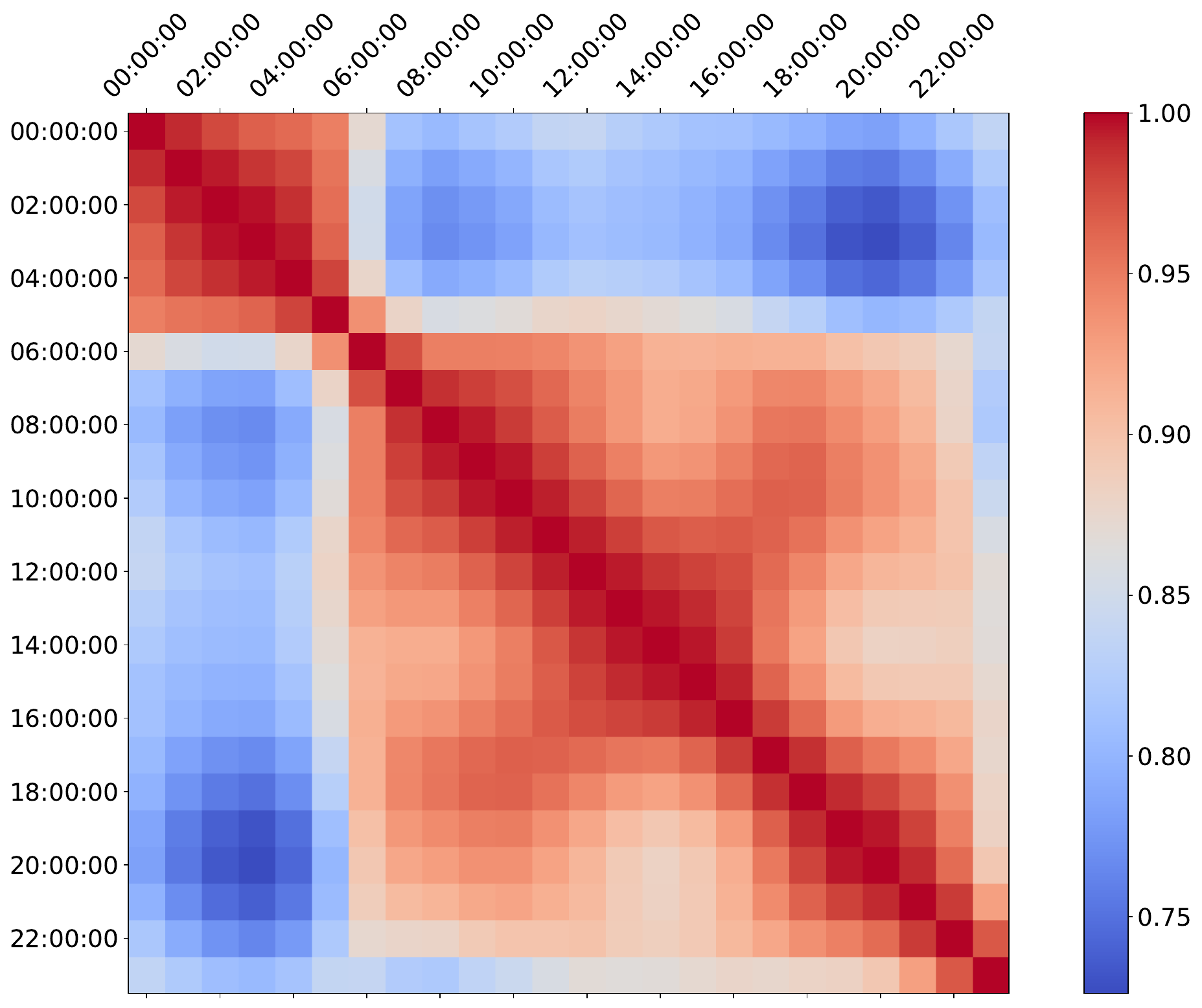}
        \caption{Correlation}
        \label{fig:Corr}
    \end{subfigure}%
    \hfill
    \begin{subfigure}[b]{0.48\textwidth}
        \centering
        \includegraphics[width=\textwidth,height=4cm]{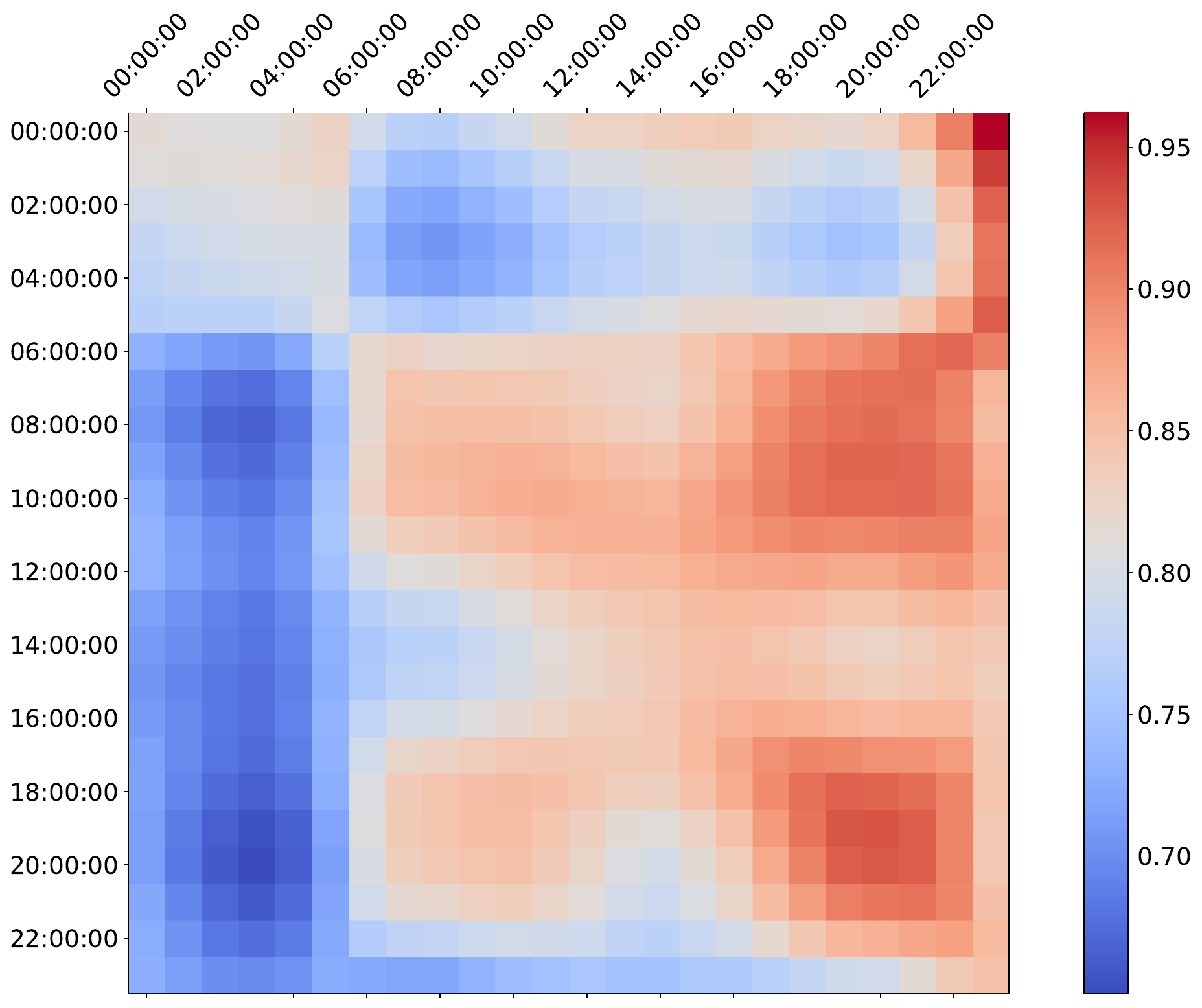 }
        \caption{Lag 1 correlation}
        \label{fig:Corr1}
    \end{subfigure}
    
    \vspace{0.3cm}
    \begin{subfigure}[b]{0.48\textwidth}
        \centering
        \includegraphics[width=\textwidth,height=4cm]{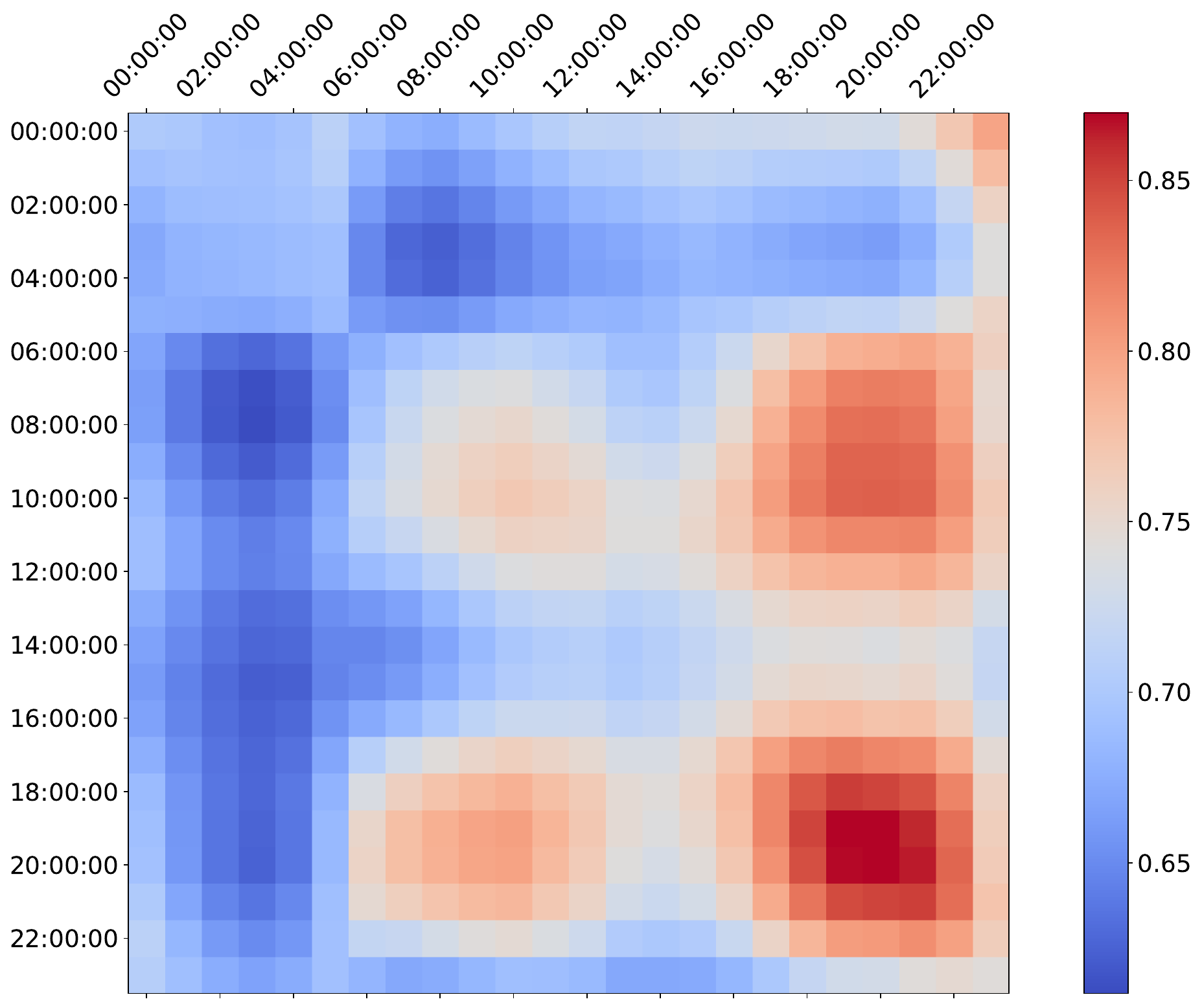}
        \caption{Lag 2 correlation}
        \label{fig:Corr2}
    \end{subfigure}%
    \hfill
    \begin{subfigure}[b]{0.48\textwidth}
        \centering
        \includegraphics[width=\textwidth,height=4cm]{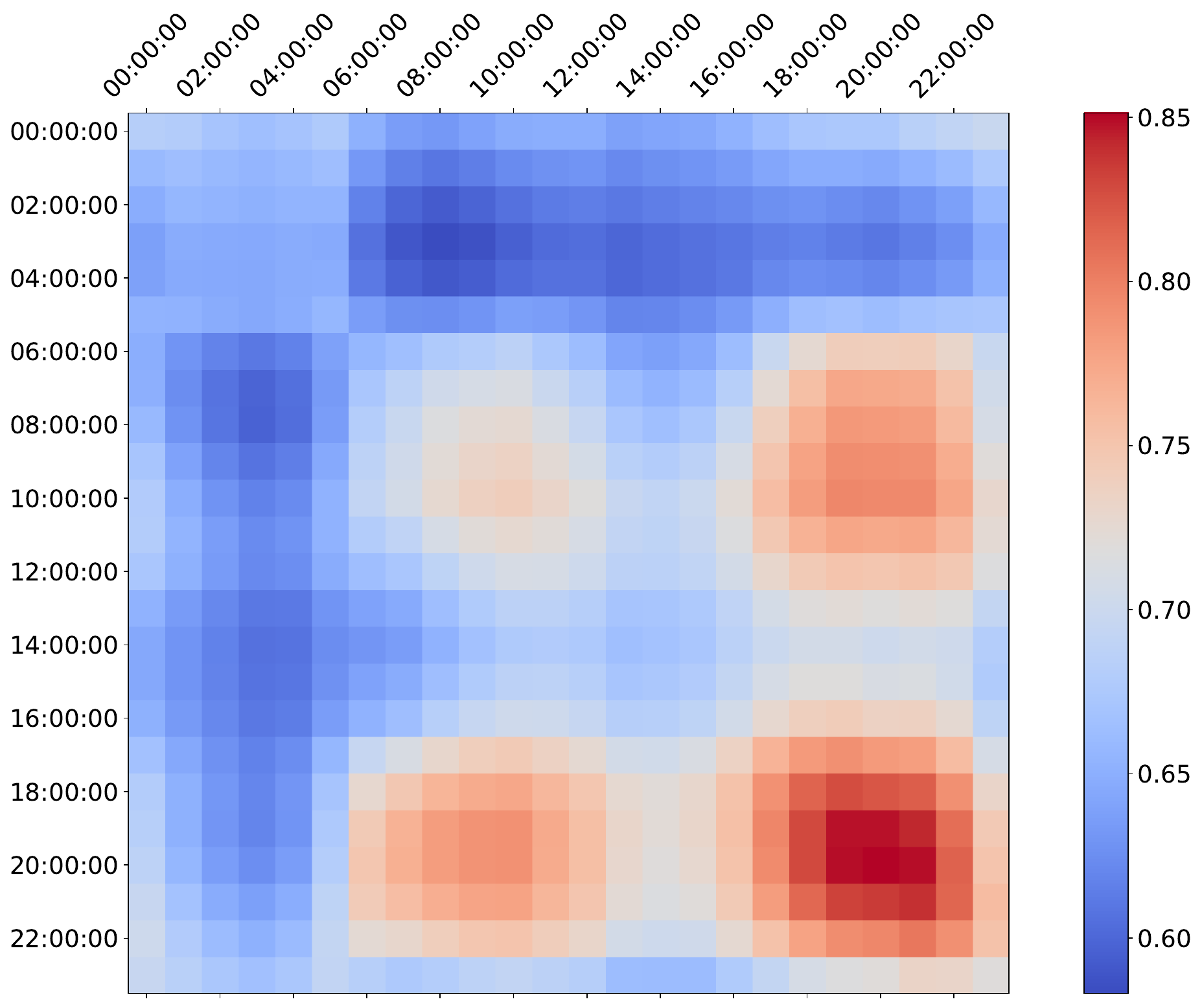}
        \caption{Lag 3 correlation}
        \label{fig:Corr3}
    \end{subfigure}
    
    \vspace{0.3cm}
    \begin{subfigure}[b]{0.48\textwidth}
        \centering
        \includegraphics[width=\textwidth,height=4cm]{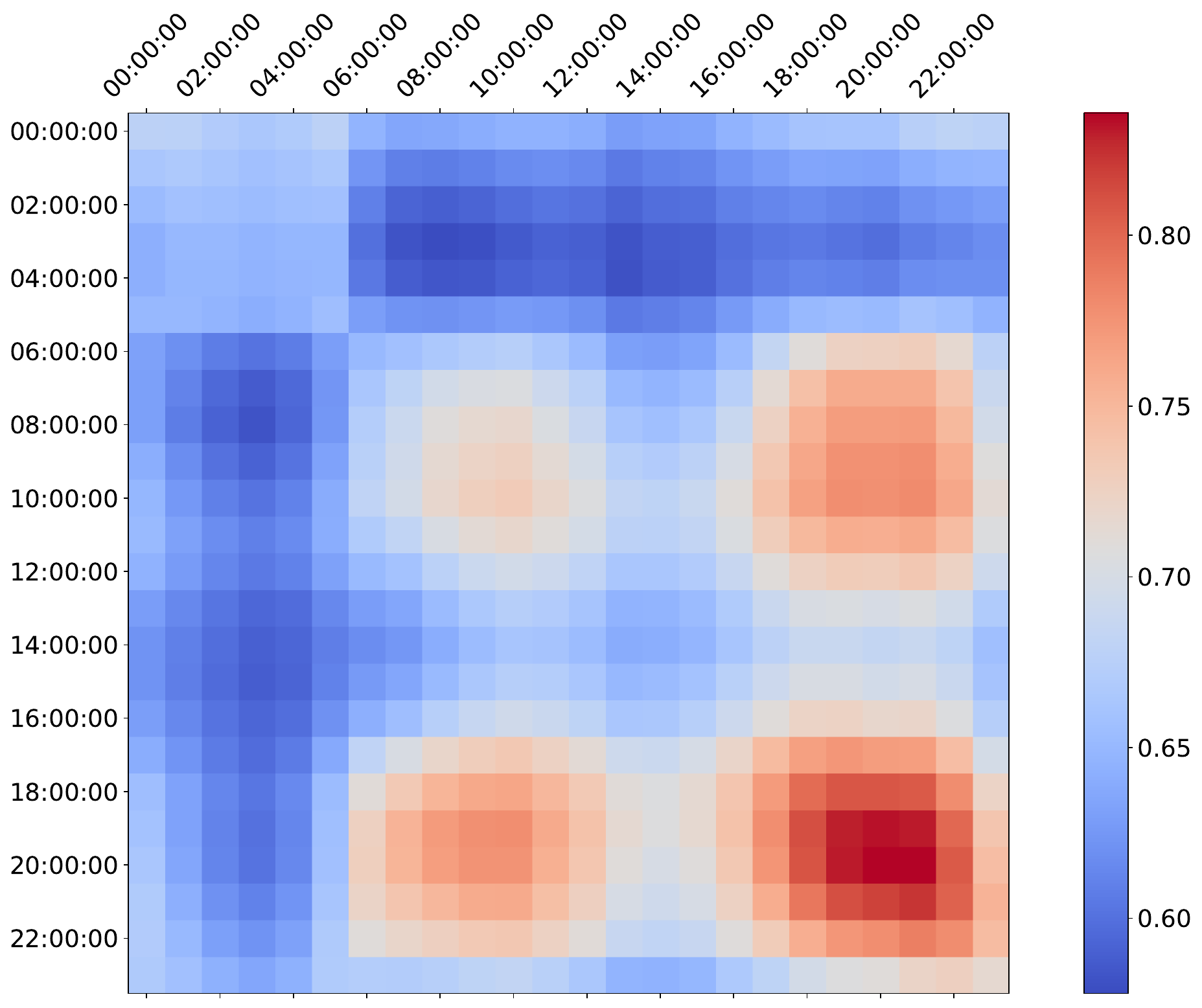}
        \caption{Lag 4 correlation}
        \label{fig:Corr4}
    \end{subfigure}%
    \hfill
    \begin{subfigure}[b]{0.48\textwidth}
        \centering
        \includegraphics[width=\textwidth,height=4cm]{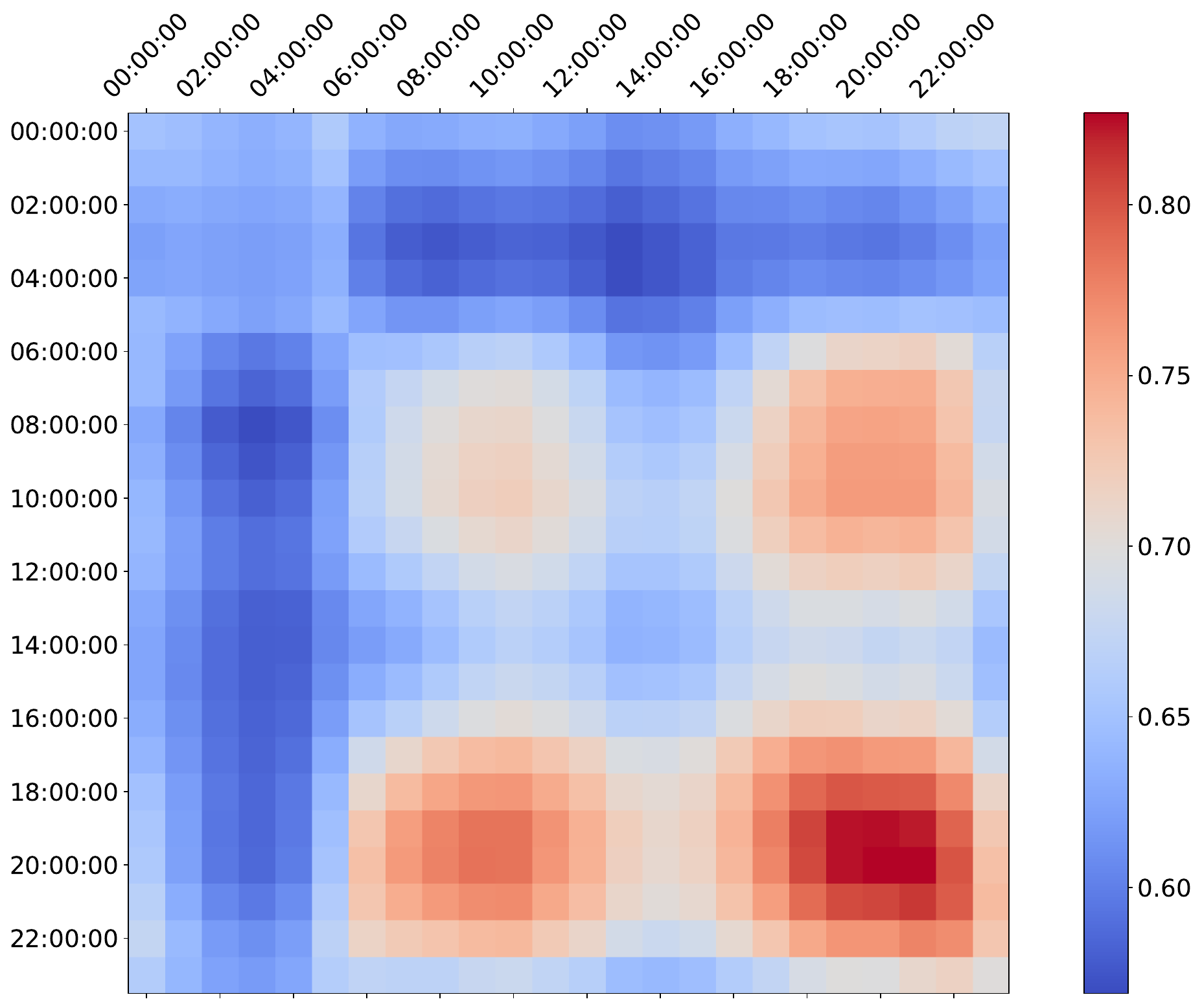}
        \caption{Lag 5 correlation}
        \label{fig:Corr5}
    \end{subfigure}
    
    \vspace{0.3cm}
    \begin{subfigure}[b]{0.48\textwidth}
        \centering
        \includegraphics[width=\textwidth,height=4cm]{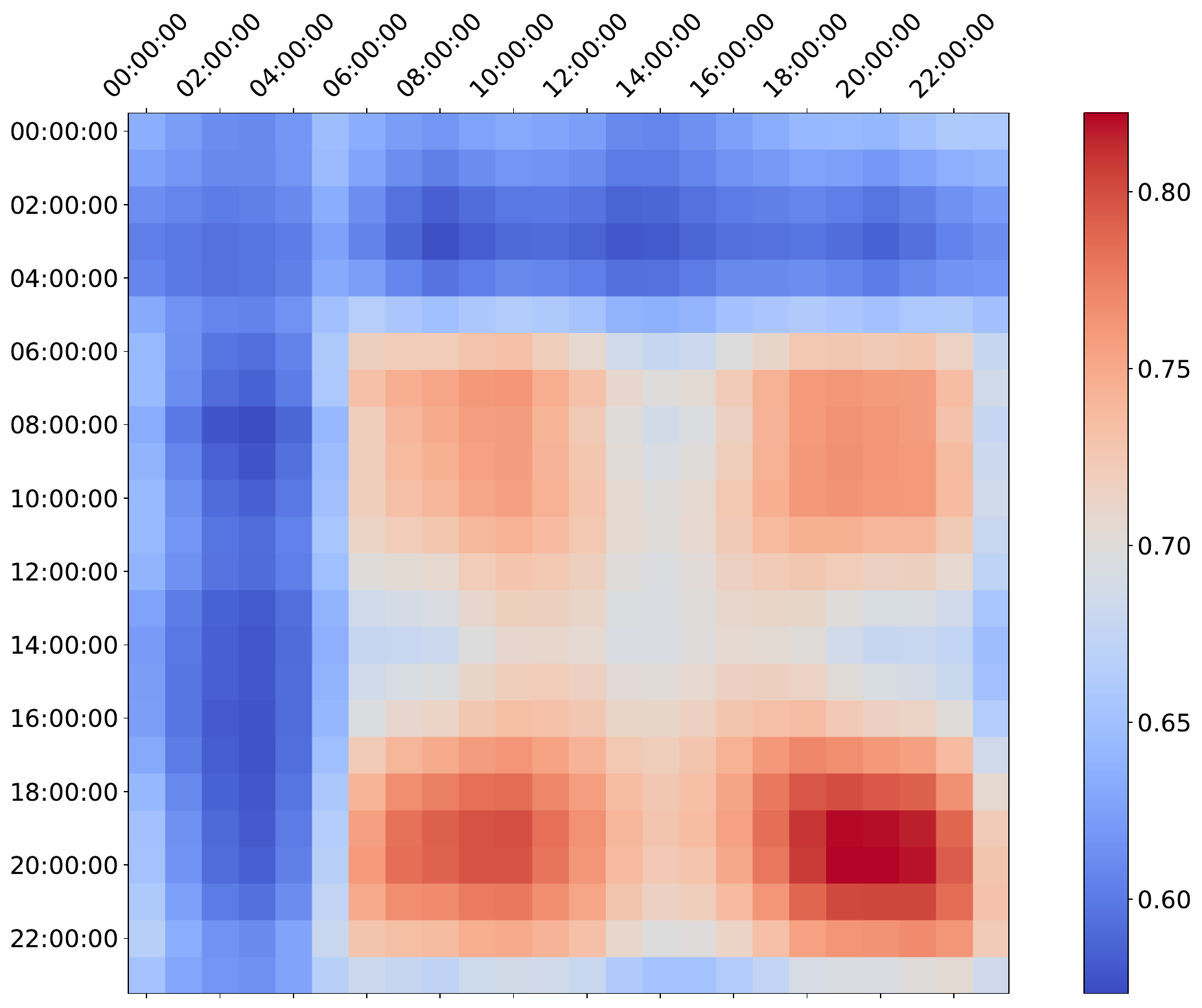}
        \caption{Lag 6 correlation}
        \label{fig:Corr6}
    \end{subfigure}%
    \hfill
    \begin{subfigure}[b]{0.48\textwidth}
        \centering
        \includegraphics[width=\textwidth,height=4cm]{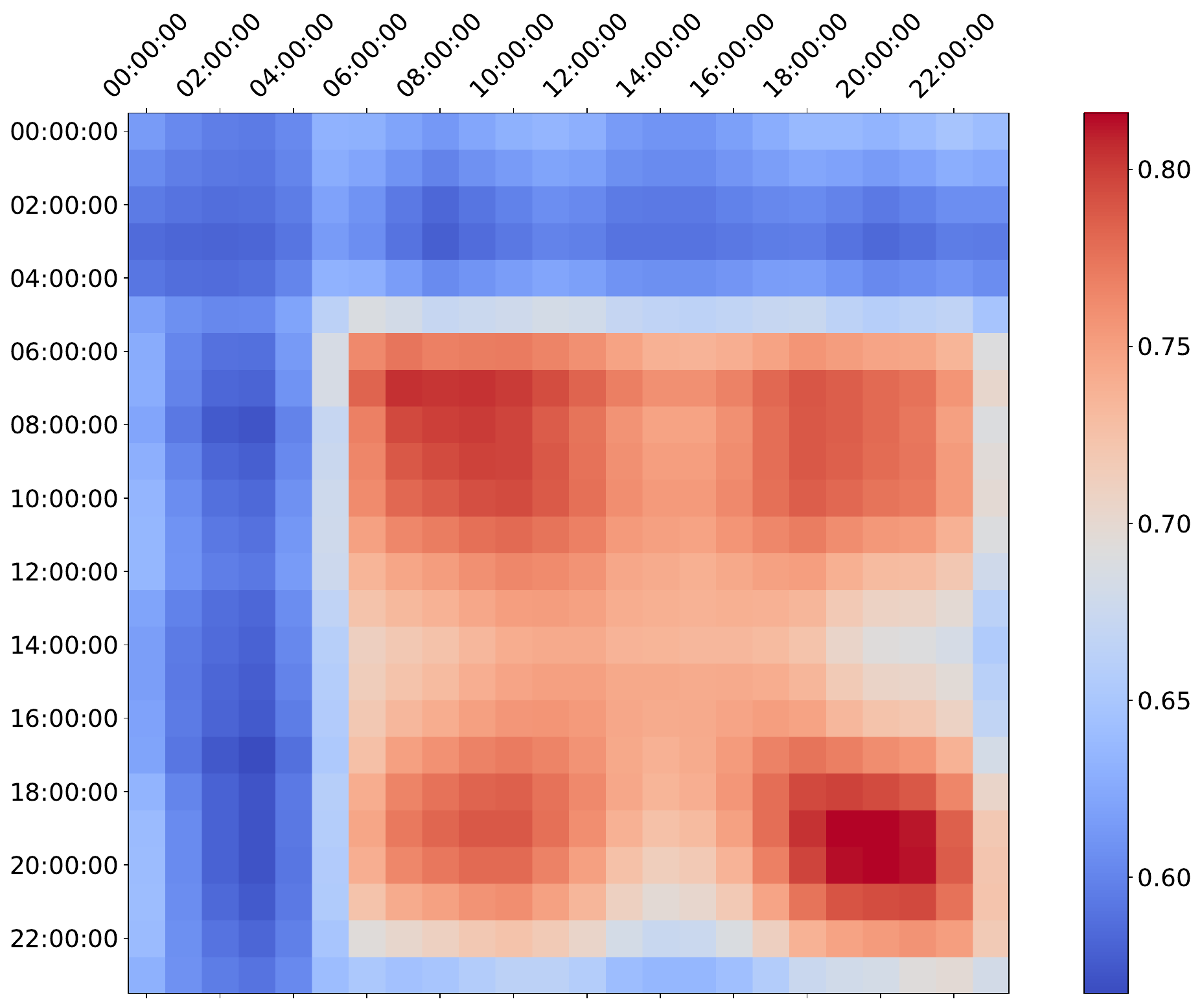}
        \caption{Lag 7 correlation}
        \label{fig:Corr7}
    \end{subfigure}
    
    \caption{Correlation matrices of hourly prices, with a lag of 0 to 7 days.}
    \label{fig:Correlation Structure}
\end{figure}

\subsection{Stationarity Analysis}
\begin{figure}
\centering
\includegraphics[height=0.5\columnwidth]{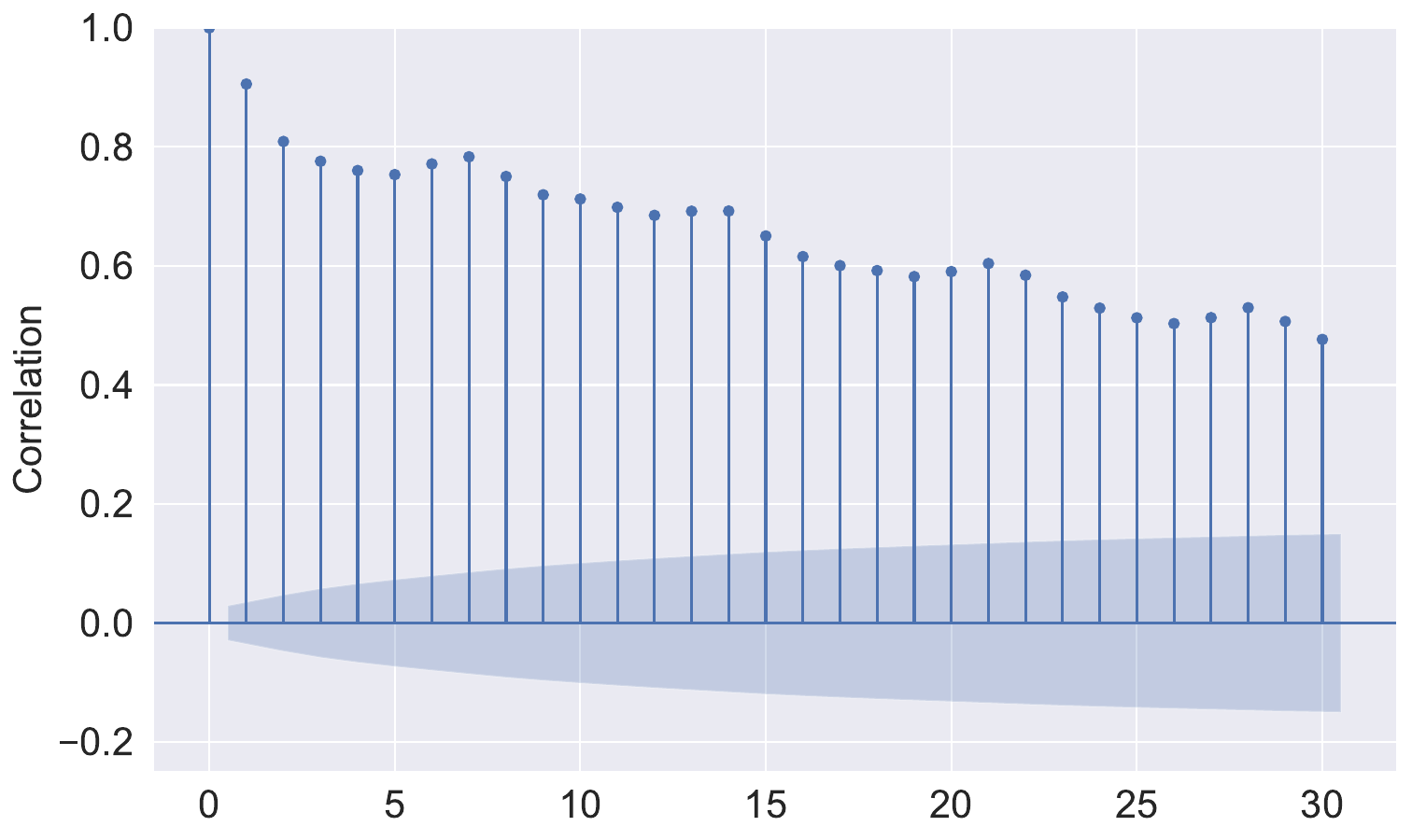}
\caption{ACF of the daily average of the Nordic SP.}
\label{fig:Weekly seasonality}
\end{figure}

We first applied the Augmented Dickey–Fuller (ADF) test~\cite{said1984testing,dickey1979distribution} in its three standard variants (no constant or trend, constant only, constant and trend) to the Nordic SP for each delivery period. In all cases, the null hypothesis of a unit root was rejected, suggesting stationarity.

To complement this, we also applied the Kwiatkowski–Phillips–Schmidt–Shin (KPSS) test~\cite{kwiatkowski1992testing}, which uses stationarity as the null hypothesis and tests for trend-stationarity. This distinction is crucial: a process can have no unit root and still be non-stationary but trend-stationary. When a shock occurs, a trend-stationary process reverts to its trend, whereas a unit-root process experiences a permanent mean shift. Using both the constant-only and trend-stationary versions of the KPSS test, we rejected the null of stationarity for all delivery periods. Thus, the ADF and KPSS tests yield contradictory conclusions.

Standard ADF and KPSS tests do not assess stationarity at seasonal frequencies. Figure~\ref{fig:Weekly seasonality} shows the autocorrelation function (ACF) of the daily average Nordic SP, revealing significant correlations at lags of multiples of seven days (weekly seasonality). To formally test for seasonal unit roots, we applied the Canova–Hansen (CH) test~\cite{canova1995seasonal}\footnote{A concise mathematical background is provided in appendix \ref{app:1}.}, which uses seasonal stationarity as the null and employs nonparametric methods to detect general seasonal patterns. The CH test indicated the presence of a weekly seasonal unit root in all delivery periods except the first four.

We conjecture that the apparent ADF–KPSS contradiction is due to these weekly seasonal unit roots. To validate this, we re-applied both tests after weekly differencing and confirmed stationarity across all delivery periods. The exceptions in the first four periods are likely due to the structural break observed in late 2021, which affects the mean but not the seasonal component.

While strict stationarity is not required for neural network forecasting, it is critical for consistent coefficient estimation in statistical models such as VAR. Accordingly, we use weekly-differenced series for all statistical models to ensure valid inference and mitigate the risk of spurious regression.

\section{Forecast-Optimized Feature Engineering Approach}
\label{sec:Section IV}
Our proposed feature-engineering framework is designed to produce input features that are both interpretable and optimized for predictive performance. It consists of three key steps. First, we perform interpretable feature selection by identifying explanatory variables for which a shock is transmitted to the system price (SP), allowing market participants to understand which drivers truly matter. Second, we apply principal component analysis (PCA) to the selected variables and the SP to mitigate imperfect multicollinearity, ensuring a stable and well-conditioned input space. Finally, rather than relying on heuristic approaches such as the elbow method, we integrate PCA with the downstream forecasting task by selecting the number of components that minimizes the root mean squared error (RMSE) of the forecast, as detailed in Section~\ref{sec:multi-forecast}. This approach explicitly links feature engineering to forecast accuracy, ensuring that regularization serves the ultimate goal of improving predictive performance.

\subsection{Interpretable feature selection}
\label{subsec:Interpretable feature engineering}
In this subsection, we present our approach for obtaining the feature vector $\textbf{x}$ from the raw feature vector $\textbf{y}$. As raw features, i.e., columns of the $\textbf{Y}$ matrix, we use the production categories shown in Table~\ref{tab:Production Profile} (15 raw features), daily system volumes, i.e. consumption (1 raw feature), auction capacities for import and export (2 raw features) and the natural gas price  (1 raw feature), i.e., in total 19 raw features. 

The novelty of our feature selection approach is twofold. First, motivated by the intuition that the main drivers may vary with the price level, we cluster the SP before selecting features. Second, we identify features based on whether shocks to them are transmitted to the SP. In addition, we remove seasonal and autocorrelation effects. This is a crucial step for recognizing actual drivers of the system price, since neglecting this filtering process will generate fake correlations between variables. 

We use the daily volume-weighted average of the 24 hourly system prices $\overline{\textbf{p}}_d=\frac{\textbf{c}_d.\textbf{p}_d}{\textbf{c}_d.\textbf{1}}$ as a representative of the daily Nordic SP, where $\textbf{c}_d = (c_{1d},...,c_{24d})$ is the hourly system volume (consumption) vector on day $d$ (MWh).

To identify the main drivers at different price levels, we propose to cluster the Nordic SP based on the daily volume-weighted average SP, using K-means clustering. We used the elbow method to identify the number of clusters, resulting in $K=3$ clusters in terms of price, corresponding to low, moderate, and high price levels, as shown in Fig.~\ref{fig:Price_Clusters}.

\begin{figure}
\centering
\includegraphics[width=0.7\columnwidth]{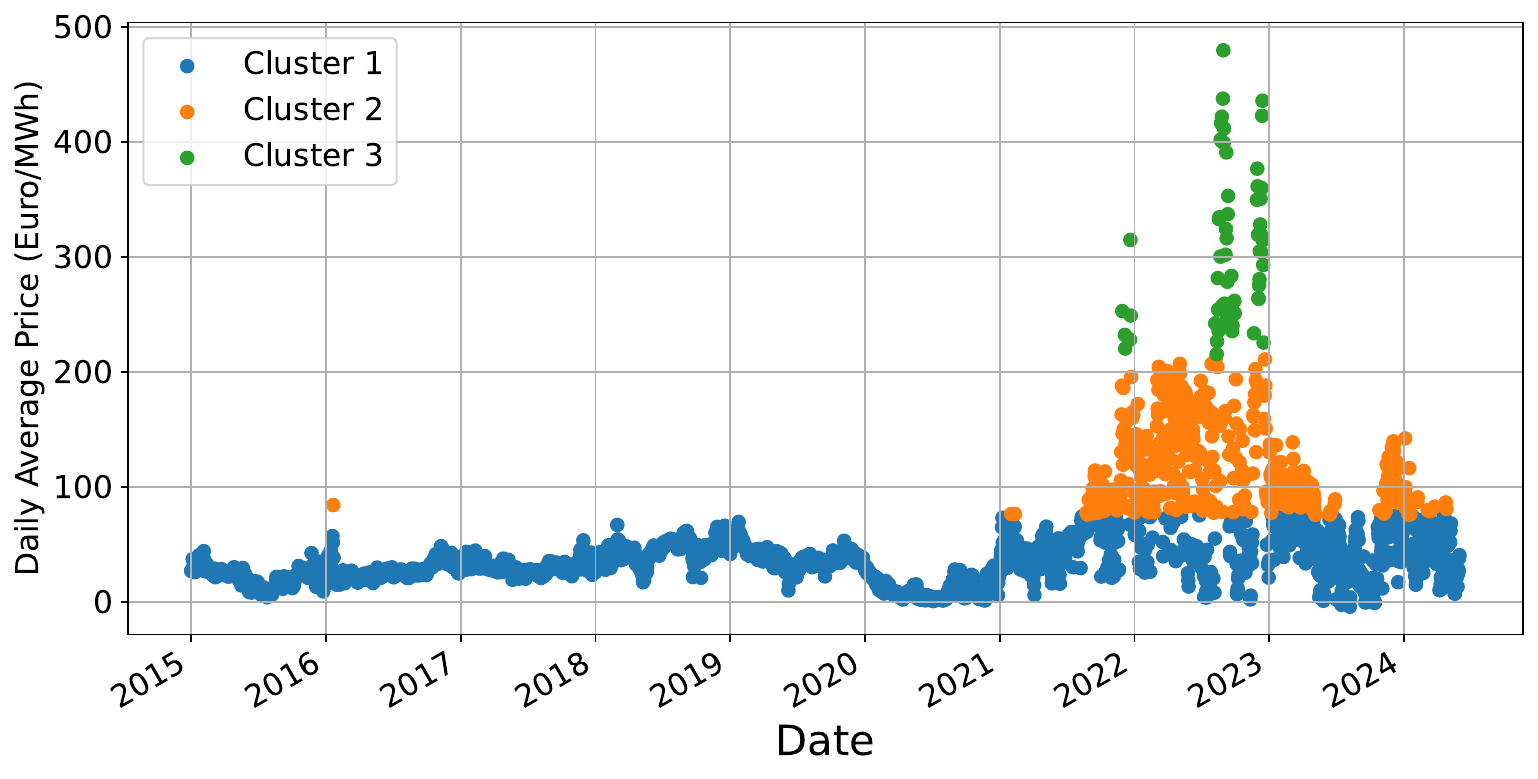}
\caption{Price clusters of the Nordic SP for $K=3$.}
\label{fig:Price_Clusters}
\end{figure}

\begin{table}
\centering
\caption{The Nordic production profile in 2023}
\label{tab:Production Profile}
\resizebox{\textwidth}{!}{  
\begin{tabular}{c c c c c c c c c c c c c c c c c c c c c c c c}
\toprule
& Fossil Peat & Hydro Reservoir & Gas & Nuclear & Wind Onshore & Waste & Fossil Oil & Biomass \\
\midrule
Total (TWh) & 2.27 & 172.56 & 5.2 & 79.26 & 73.22 & 1.32 & 0.33 &9.06\\
Share \%  & 0.55 & 41.67 & 1.26 & 19.14 & 17.68 & 0.32 & 0.08 & 2.19  \\
Variation \%  & 75.64 & 26.06 & 41.06 & 14.68 & 49.56& 17.84 &34.91 &38.97 \\
\midrule
& Wind Offshore  & River And Poundage & Solar & Other Renewable & Pumped Storage & Hard Coal & Other\\
\midrule
Total (TWh) & 8.29 & 42.25 & 4.94 & 0.46 & 1.75 & 5.27 & 7.89 \\
Share \% & 2.0 & 10.2 & 1.19 & 0.11 & 0.42 & 1.27 & 1.91  \\
Variation \% & 58.51 & 18.2 & 82.72 &29.73 & 103.56 &76.76 & 51.42 \\
\bottomrule
\end{tabular}
}
\end{table}

Table~\ref{tab:Production Profile} summarizes the production profile for the Nordic region in 2023, and Fig.~\ref{fig:Consumption-Production} shows the daily electricity consumption and production by major category. The main production sources are \textit{Hydro Water Reservoir}, \textit{Nuclear}, \textit{Wind Onshore}, and \textit{Hydro Run-of-River and Poundage}, while all other categories each contribute less than 3\% of total production. The third row of the table reports the coefficient of variation for daily production, revealing substantial volatility across all categories. Nuclear generation is the most stable, followed by waste, hydro run-of-river, and poundage.  

Because these variables exhibit both seasonal and non-seasonal autoregressive patterns, we first remove seasonality and cyclical components from the daily volume-weighted average SP and from the 19 raw features using the MSTD algorithm~\cite{bandara2025mstl}. The MSTD is well suited for time series with multiple seasonalities, such as the Nordic SP, and produces an additive decomposition into trend, seasonal, and residual components.  

We retain the residual-plus-trend components for each time series and fit a SARIMA model to remove remaining autocorrelation. We then keep the SARIMA residuals and, within each SP cluster, regress the SP residuals on the feature residuals. For all features except import and export auction capacities, we use lagged values from day~$d-1$; for import and export capacities, we use day~$d$ values. This process yields de-seasonalized and de-autocorrelated drivers, allowing us to isolate the statistically significant relationships between SP and its underlying production features.

\begin{figure}[htbp]
    \centering
    \subfloat[Daily Consumption]{%
        \includegraphics[width=0.3\textwidth, height=3cm]{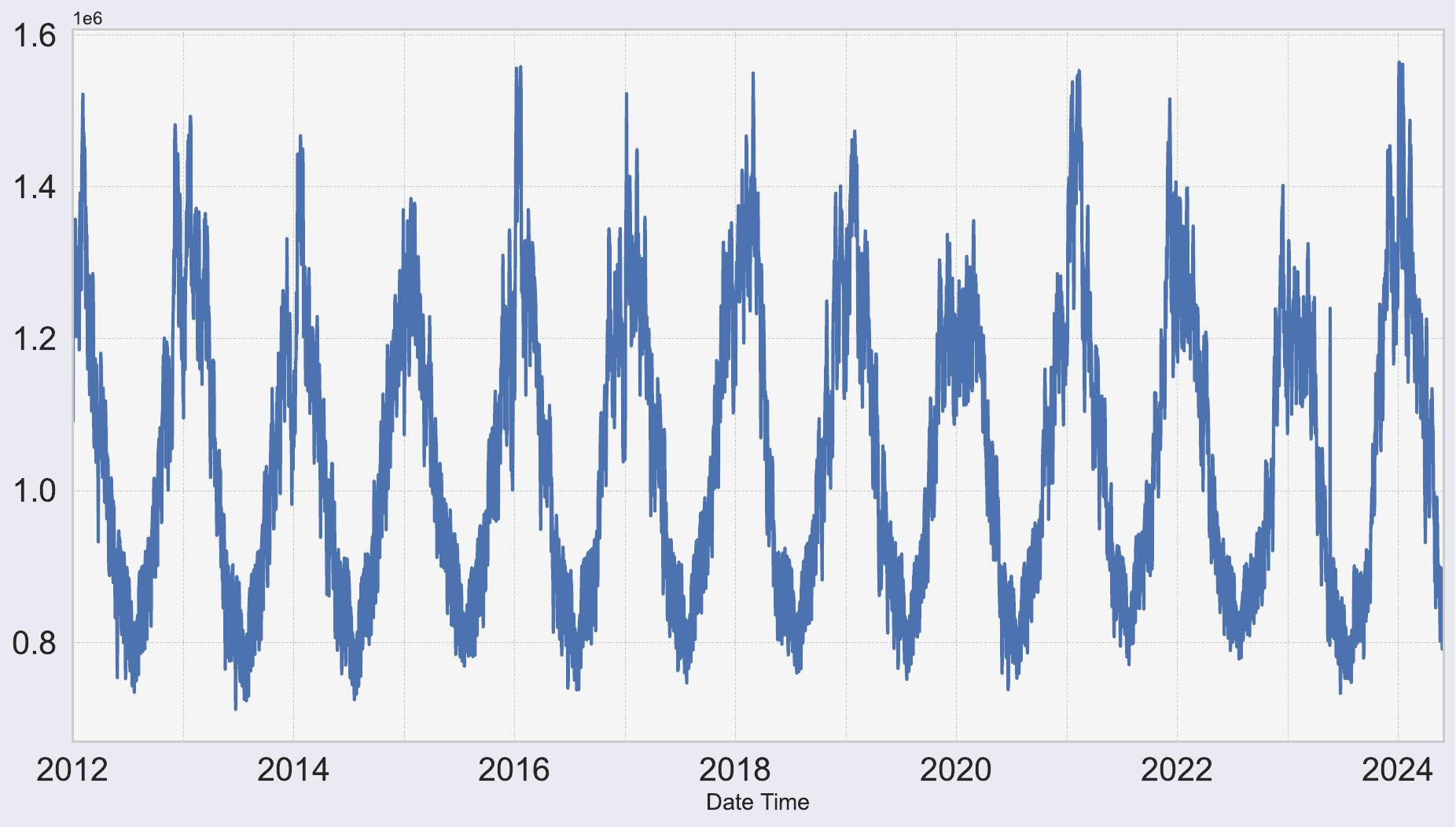}%
        \label{fig:Consumption}}
    \hfill
    \subfloat[Daily Production]{%
        \includegraphics[width=0.3\textwidth, height=3cm]{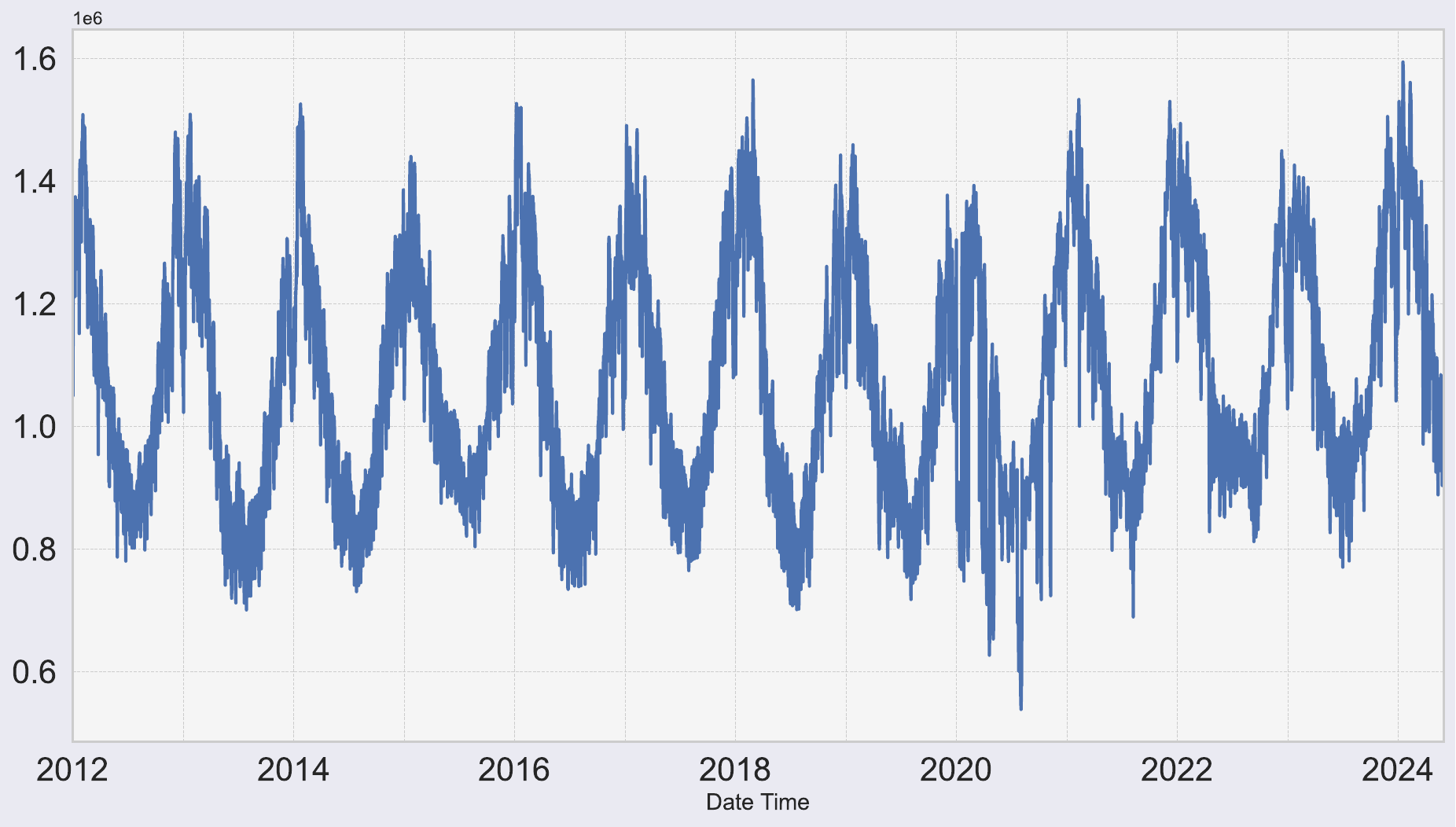}%
        \label{fig:Production}}
    \hfill
    \subfloat[Reservoir Production]{%
        \includegraphics[width=0.3\textwidth, height=3cm]{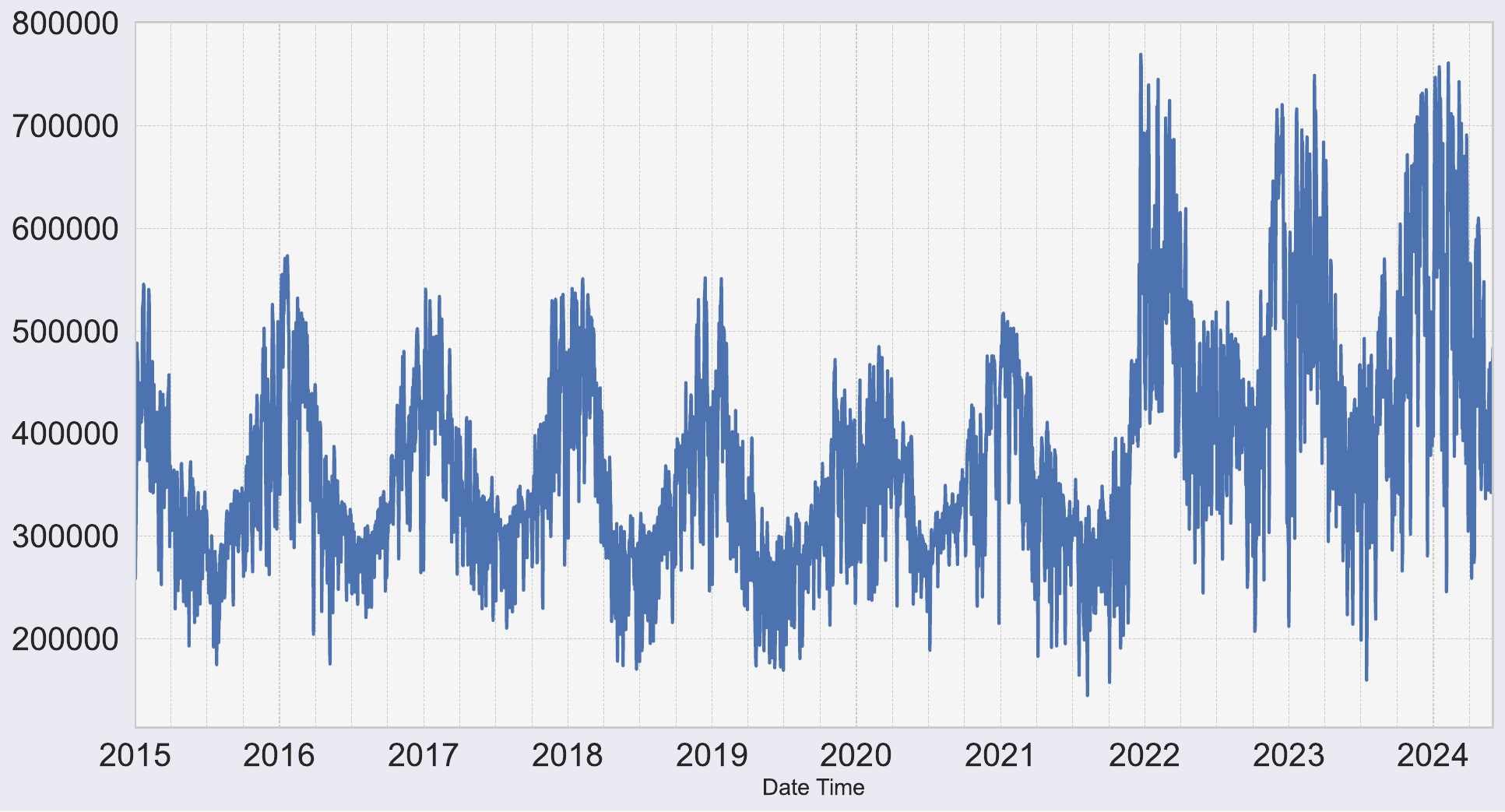}%
        \label{fig:Reservoir}}
    
    \subfloat[Wind Onshore Production]{%
        \includegraphics[width=0.3\textwidth, height=3cm]{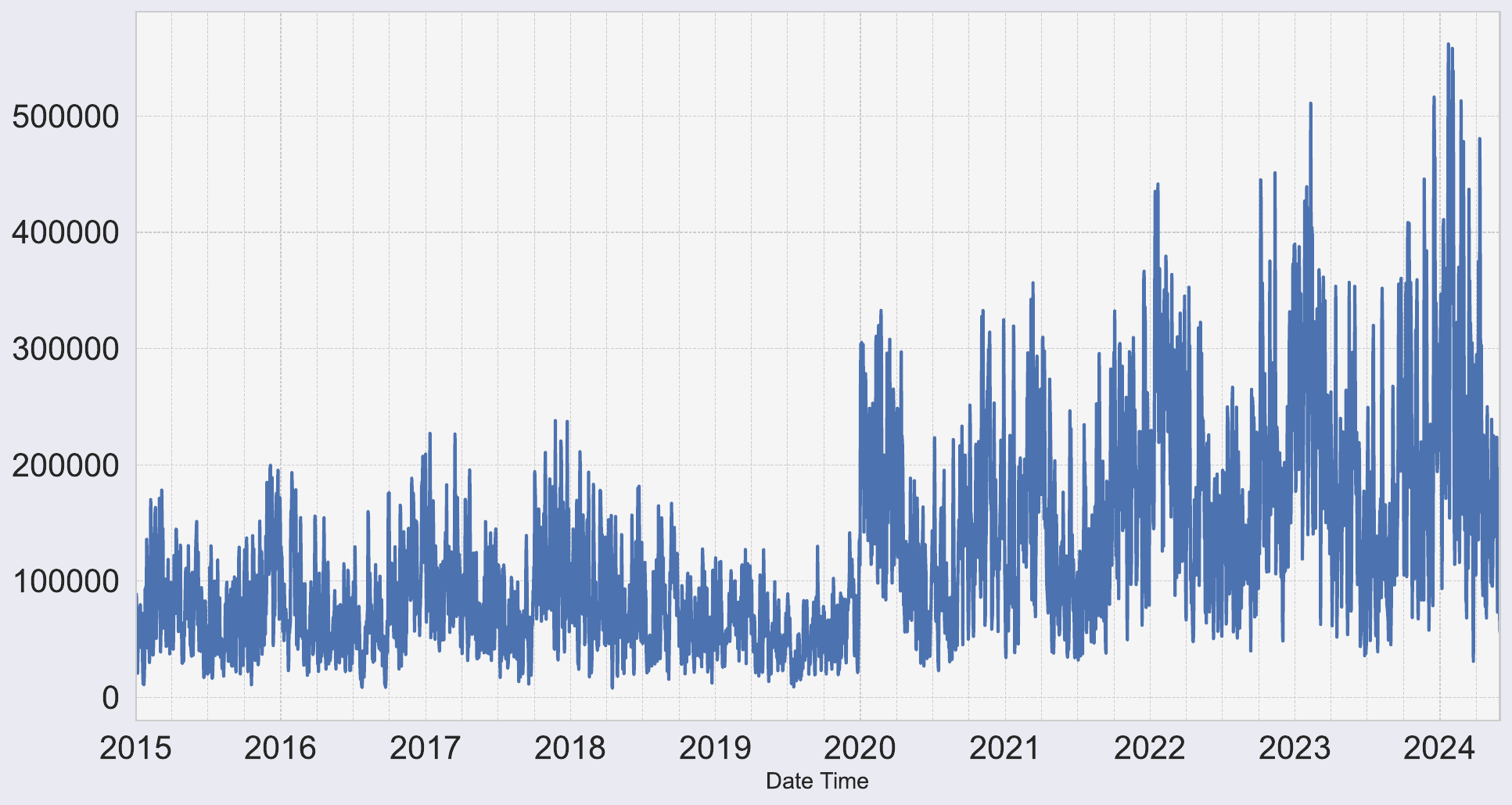}%
        \label{fig:Wind_Onshore}}
    \hfill
    \subfloat[Nuclear Production]{%
        \includegraphics[width=0.3\textwidth, height=3cm]{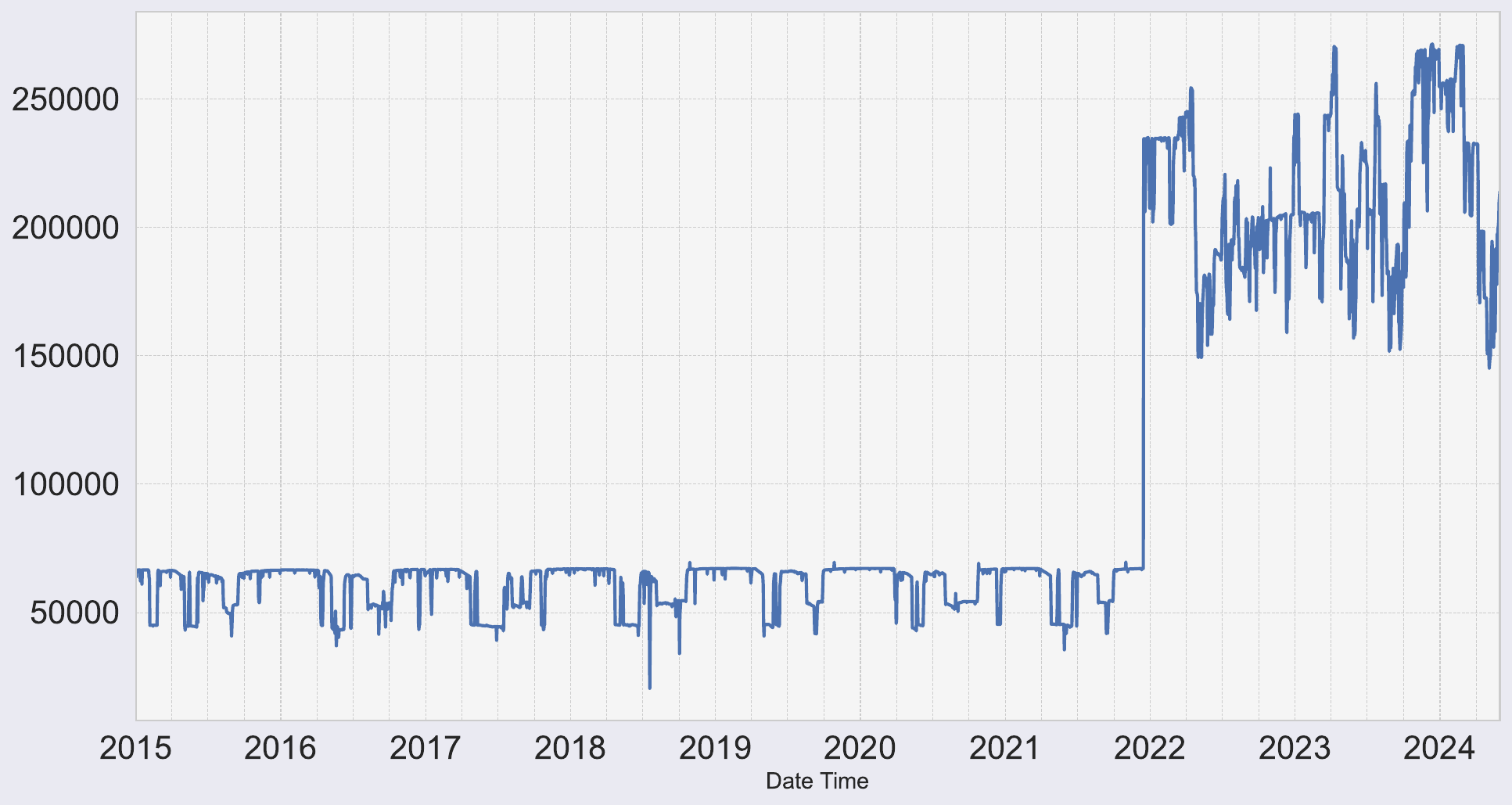}%
        \label{fig:Nuclear}}
    \hfill
    \subfloat[River and Poundage Production]{%
        \includegraphics[width=0.3\textwidth, height=3cm]{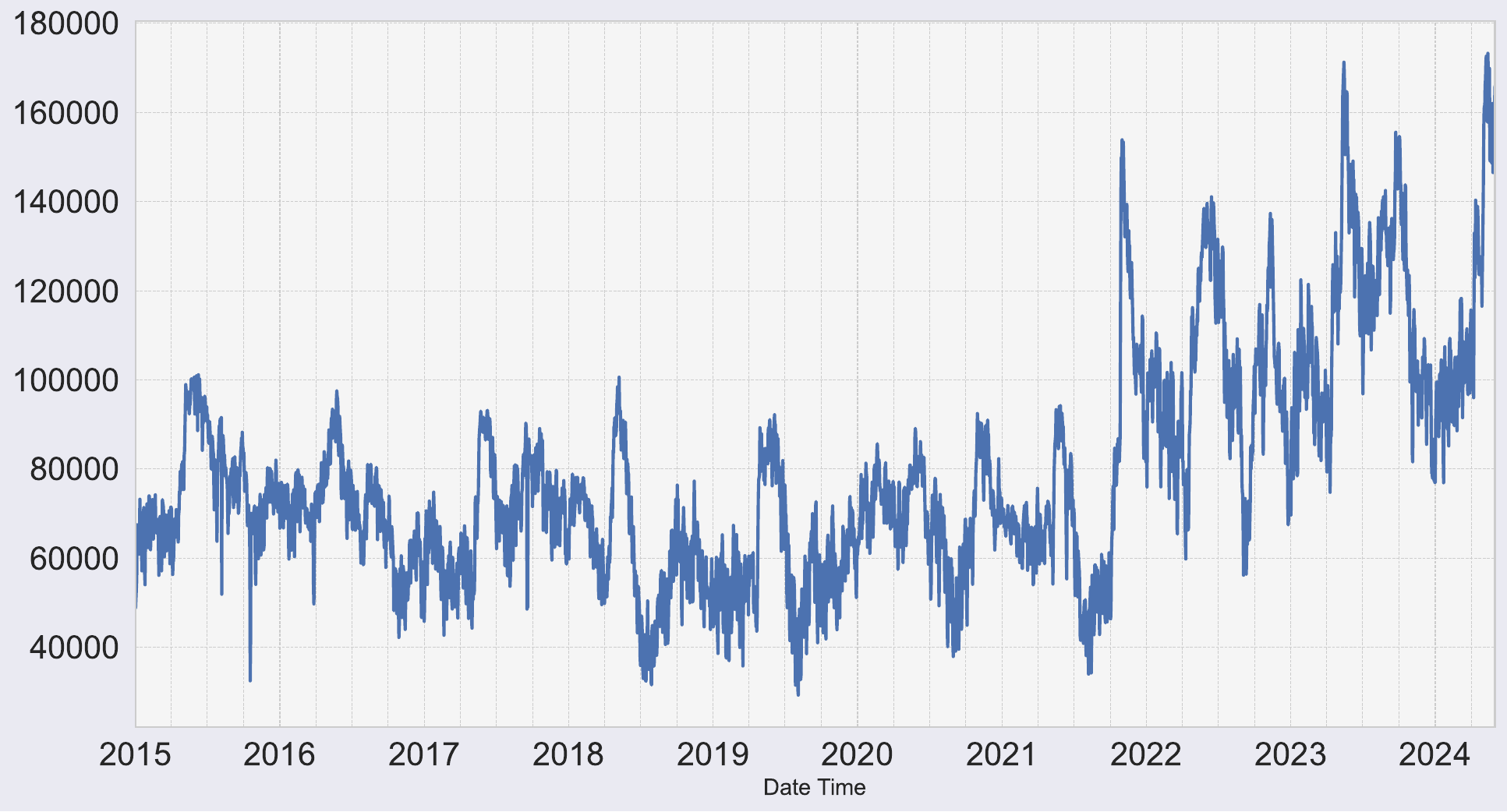}%
        \label{fig:RiverPoundage}}
    
    \subfloat[Fossil Gas Production]{%
        \includegraphics[width=0.3\textwidth, height=3cm]{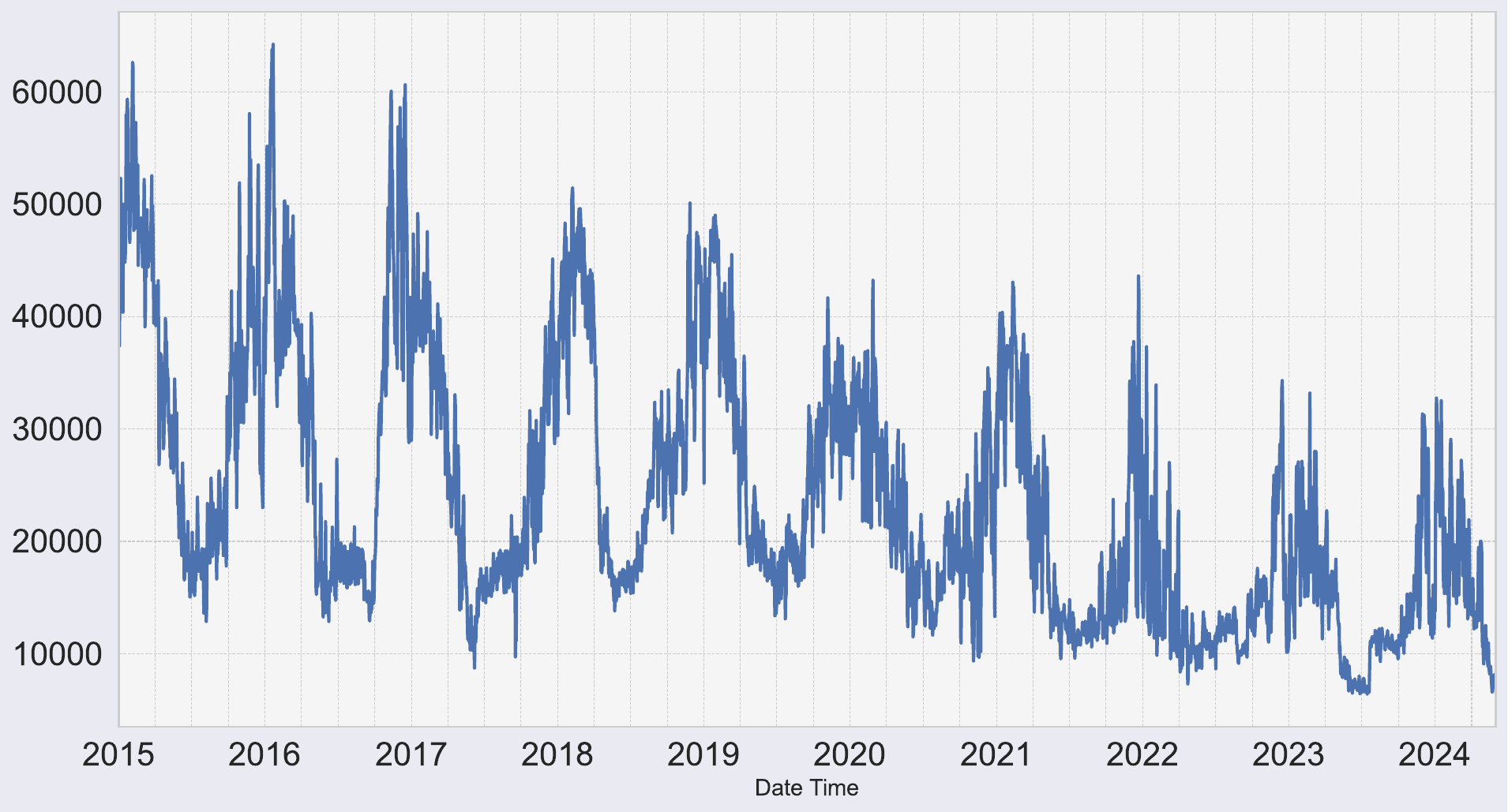}%
        \label{fig:GAS}}
    \hfill
    \subfloat[Biomass Production]{%
        \includegraphics[width=0.3\textwidth, height=3cm]{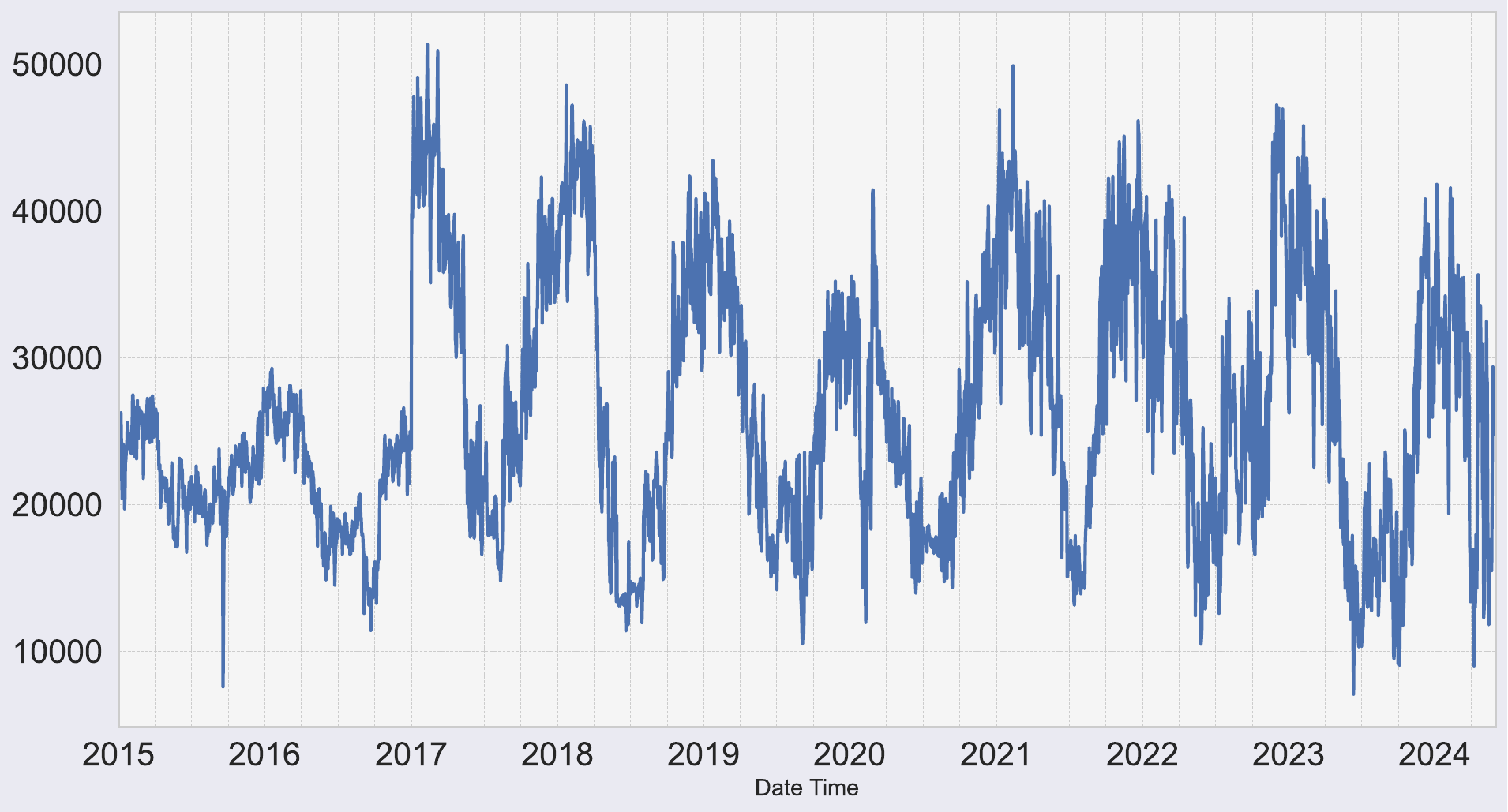}%
        \label{fig:Biomass}}
    \hfill
    \subfloat[Solar Production]{%
        \includegraphics[width=0.3\textwidth, height=3cm]{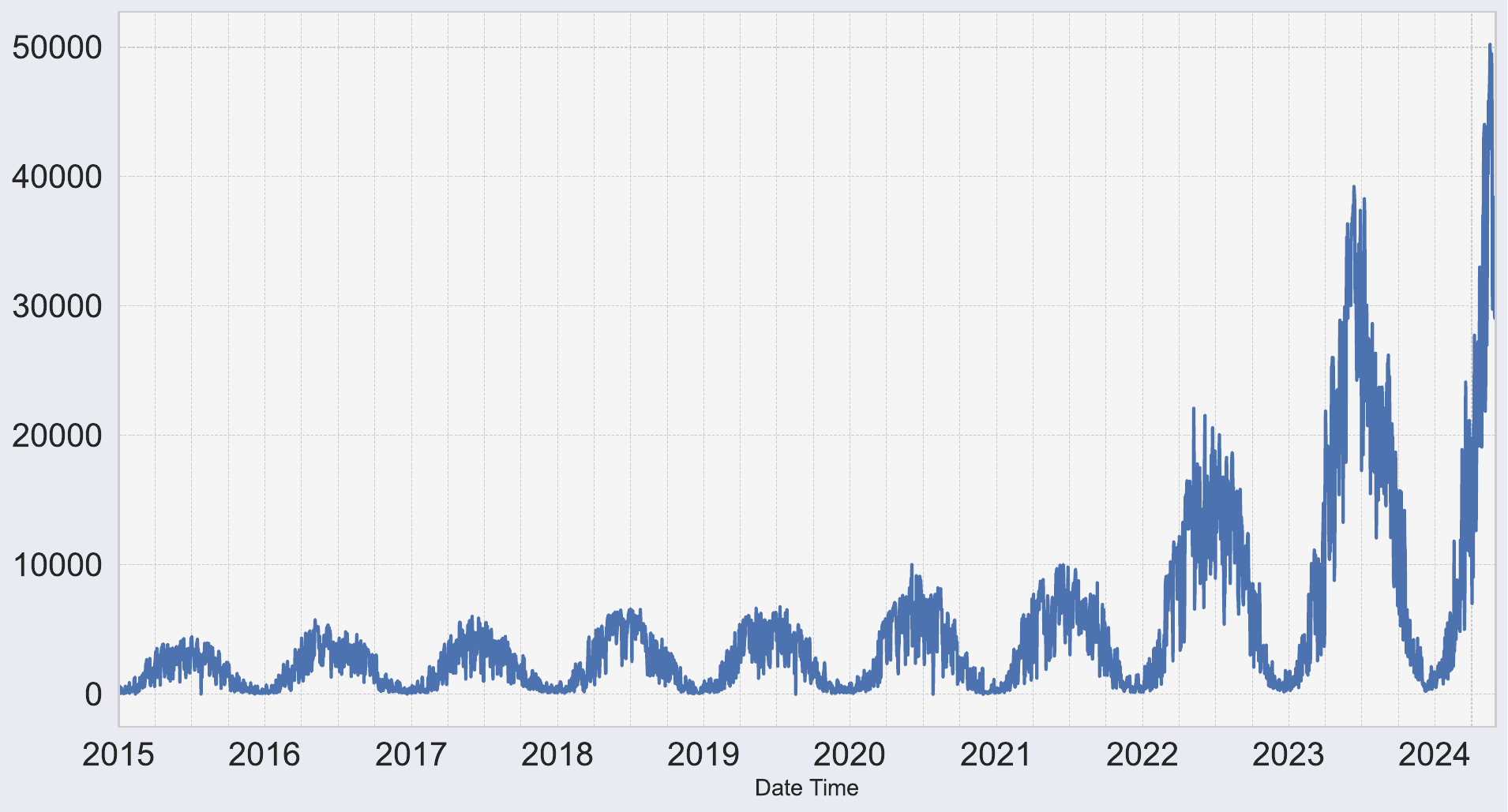}%
        \label{fig:Solar}}
    
    \caption{Consumption and production categories (MWh).}
    \label{fig:Consumption-Production}
\end{figure}

After obtaining the linear regression coefficients for all raw features, we kept those with a $p$ value less than $5$\%, and we refer to these as \textbf{\textit{features}}. The proposed feature-engineering algorithm is shown in Fig. \ref{fig:feature_engineering}. 

Our results show that the set of retained features is cluster-dependent. Table~\ref {tab:ols_results_cluster_First} shows that in the first cluster only Biomass, Nuclear, Hydro Water Reservoir, Wind Onshore, Fossil Gas, Other (non-categorized production), and Auction Capacity for Import are significant at $5$\% level. This leads us to the following observation.

\begin{figure}[H]
    \centering
    \resizebox{\textwidth}{!}{\begin{tikzpicture}[
    node distance=1cm,
    startstop/.style={ellipse, minimum width=3cm, minimum height=1cm, text centered, draw=black},
    process/.style={rectangle, rounded corners, minimum width=3cm, minimum height=1cm, text centered, draw=black},
    io/.style={trapezium, trapezium left angle=70, trapezium right angle=110, minimum width=1.2cm, minimum height=0.7cm, text centered, draw=black},
    arrow/.style={-stealth, thick}
]

\node (vwasp) [startstop] {VWASP};
\node (raw_features) [startstop, right=4cm of vwasp] {Raw features};
\node (kmeans) [process, below=1cm of vwasp] {K-means clustering};
\node (price_cluster1) [io, below=1.5cm of kmeans, xshift=-2.5cm] {Price Cluster 1};
\node (dots1) [below=1.8cm of kmeans] {...};
\node (price_clusterK) [io, below=1.5cm of kmeans, xshift=2.5cm] {Price Cluster K};
\node (regress) [process, below=1.75cm of dots1, minimum width=6cm] {Regression of price on raw features in different clusters};
\node (sf_cluster1) [io, below=1.5cm of regress, xshift=-2.5cm] {S.F. Cluster 1};
\node (dots2) [below=1.8cm of regress] {...};
\node (sf_clusterK) [io, below=1.5cm of regress, xshift=2.5cm] {S.F. Cluster K};
\node (mstd) [process, below=1cm of raw_features] {MSTD};
\node (tr) [io, below=1.5cm of mstd] {$  T+R  $};
\node (sarima) [process, below=1.5cm of tr] {SARIMA};
\node (residuals) [io, below=1.62cm of sarima] {Residuals};


\draw [arrow] (raw_features) -- (mstd);

\coordinate (lane_vm) at ($(vwasp.east)!0.5!(raw_features.west)$);
\draw [arrow] (vwasp.east) -- (lane_vm)
                      -- (lane_vm |- mstd.west)
                      -- (mstd.west);

\draw [arrow] (mstd) -- (tr);
\draw [arrow] (tr) -- (sarima);
\draw [arrow] (sarima) -- (residuals);

\coordinate (resid_start) at (residuals.north west);               
\coordinate (q1) at ($(resid_start)+(0,0.8)$);                      
\coordinate (q2) at ($(q1 -| regress.south east)+(-0.6,0)$);        
\coordinate (target) at (regress.south -| q2);                      
\draw [arrow] (resid_start) -- (q1) -- (q2) -- (target);

\draw [arrow] (vwasp) -- (kmeans);

\coordinate (h_line_kmeans) at ($(kmeans.south)!.5!(price_cluster1.north)$);
\draw [arrow] (kmeans.south) -- (kmeans.south |- h_line_kmeans); 
\draw (h_line_kmeans -| price_cluster1.north) -- (h_line_kmeans -| price_clusterK.north); 
\draw [arrow] (h_line_kmeans -| price_cluster1.north) -- (price_cluster1.north); 
\draw [arrow] (h_line_kmeans -| price_clusterK.north) -- (price_clusterK.north);
\draw (h_line_kmeans) -- (h_line_kmeans -| dots1.north);

\coordinate (h_line_price_regress) at ($(price_cluster1.south)!.5!(regress.north)$);
\draw [arrow] (price_cluster1.south) -- (price_cluster1.south |- h_line_price_regress);
\draw [arrow] (price_clusterK.south) -- (price_clusterK.south |- h_line_price_regress);
\draw (h_line_price_regress -| price_cluster1.south) -- (h_line_price_regress -| price_clusterK.south);
\draw [arrow] (h_line_price_regress -| regress.north) -- (regress.north);

\coordinate (h_line_sf) at ($(regress.south)!.5!(sf_cluster1.north)$);
\draw [arrow] (regress.south) -- (regress.south |- h_line_sf);
\draw (h_line_sf -| sf_cluster1.north) -- (h_line_sf -| sf_clusterK.north);
\draw [arrow] (h_line_sf -| sf_cluster1.north) -- (sf_cluster1.north);
\draw [arrow] (h_line_sf -| sf_clusterK.north) -- (sf_clusterK.north);
\draw (h_line_sf) -- (h_line_sf -| dots2.north);

\end{tikzpicture}}
    \caption{The proposed interpretable feature-engineering algorithm, VWASP = Volume Weighted Average SP: $\overline{\textbf{p}}_d$, T+R: Trend and Residual components, and S.F.: Significant Features.}
    \label{fig:feature_engineering}
\end{figure}
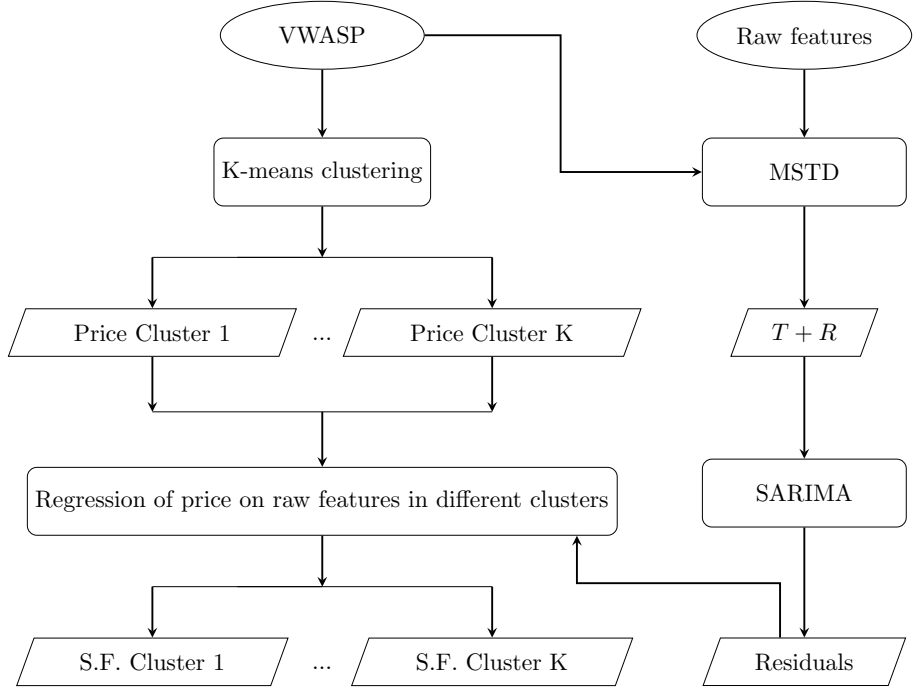

\begin{table}[t]
\caption{OLS Regression Results for First Cluster}
\centering
\setlength\tabcolsep{3.5pt} 
\renewcommand{\arraystretch}{1.15}
\small

\begin{adjustbox}{max width=\columnwidth}
\begin{tabular}{lclc}
\toprule
\textbf{Dep. Variable:}    & y & \textbf{$R^2$ (uncentered):} & 0.325 \\
\textbf{Model:}            & OLS & \textbf{Adj. $R^2$ (uncentered):} & 0.324 \\
\textbf{Method:}           & Least Squares & \textbf{F-statistic:} & 65.24 \\
\textbf{Date:}             & Mon, 21 Jul 2025 & \textbf{Prob (F-statistic):} & 1.08e-87 \\
\textbf{Time:}             & 21:33:31 & \textbf{Log-Likelihood:} & -10085. \\
\textbf{No. Observations:} & 2951 & \textbf{AIC:} & 2.018e+04 \\
\textbf{Df Residuals:}     & 2944 & \textbf{BIC:} & 2.023e+04 \\
\textbf{Df Model:}         & 7 & & \\
\textbf{Covariance Type:}  & HC3 & & \\
\bottomrule
\end{tabular}
\end{adjustbox}

\vspace{0.6em}

\begin{adjustbox}{max width=\columnwidth}
\begin{tabular}{lcccccc}
 & \textbf{coef} & \textbf{std err} & \textbf{t} & \textbf{P$>|t|$} & \textbf{[0.025} & \textbf{0.975]}  \\
\midrule
\textbf{Biomass}                   & 0.0004  & 0.000     & 3.875  & 0.000 & 0.000 & 0.001 \\
\textbf{Nuclear}                   & -0.0001 & 6.07e-05 & -2.182 & 0.029 & -0.000 & -1.35e-05 \\
\textbf{Hydro\_Water\_Reservoir}   & 0.0001  & 9.53e-06 & 15.057 & 0.000 & 0.000 & 0.000 \\
\textbf{Wind\_Onshore}             & -3.18e-05 & 6.55e-06 & -4.854 & 0.000 & -4.46e-05 & -1.9e-05 \\
\textbf{Fossil\_Gas}               & -0.0002 & 8.27e-05 & -2.369 & 0.018 & -0.000 & -3.38e-05 \\
\textbf{Other}                     & -9.91e-06 & 4.59e-06 & -2.158 & 0.031 & -1.89e-05 & -9.08e-07 \\
\textbf{Auction\_Capacity\_Import} & -0.0002 & 9.95e-05 & -1.992 & 0.046 & -0.000 & -3.23e-06 \\
\bottomrule
\end{tabular}
\end{adjustbox}

\vspace{0.6em}

\begin{adjustbox}{max width=\columnwidth}
\begin{tabular}{lclc}
\textbf{Omnibus:} & 1083.796 & \textbf{Durbin--Watson:} & 1.978 \\
\textbf{Prob(Omnibus):} & 0.000 & \textbf{Jarque--Bera (JB):} & 22375.109 \\
\textbf{Skew:} & -1.230 & \textbf{Prob(JB):} & 0.00 \\
\textbf{Kurtosis:} & 16.264 & \textbf{Cond. No.:} & 26.9 \\
\bottomrule
\end{tabular}
\end{adjustbox}

\footnotesize
\medskip

\emph{Notes:} [1] $R^2$ is computed without centering (uncentered) since the model does not contain a constant. [2] Standard errors are heteroscedasticity-robust (HC3).
\label{tab:ols_results_cluster_First}
\end{table}

\begin{observation} \label{observation1}
    When the Nordic SP is low, the auction capacity for import has explanatory power for price variations (shocks to auction capacity for import result in price shocks), but export capacity does not. The intuition is that since demand is relatively stable, any shock to this variable affects the price.
\end{observation}

By contrast, Table. \ref{tab:ols_results_Second_cluster} shows that in the second cluster only Hydro Water Reservoir and Fossil Hard Coal are significant at 5\% significance level. For the third cluster, Table. \ref{tab:ols_results_Third_cluster} shows that solar, wind, offshore, and fossil hard coal are significant. This analysis allows us to make the following observations.

\begin{table*}[t]
\centering
\caption{OLS Regression Results for Second Cluster}
\label{tab:ols_results_Second_cluster}
\setlength\tabcolsep{3.5pt}
\renewcommand{\arraystretch}{1.15}
\small

\begin{adjustbox}{max width=\textwidth}
\begin{tabular}{lclc}
\toprule
\textbf{Dep. Variable:} & y & \textbf{$R^2$ (uncentered):} & 0.522 \\
\textbf{Model:}         & OLS & \textbf{Adj.\ $R^2$ (uncentered):} & 0.520 \\
\textbf{Method:}        & Least Squares & \textbf{F-statistic:} & 175.8 \\
\textbf{Date:}          & Mon, 21 Jul 2025 & \textbf{Prob (F-statistic):} & 2.38e-56 \\
\textbf{Time:}          & 22:24:13 & \textbf{Log-Likelihood:} & -1811.4 \\
\textbf{No. Observations:} & 427 & \textbf{AIC:} & 3627 \\
\textbf{Df Residuals:}  & 425 & \textbf{BIC:} & 3635 \\
\textbf{Df Model:}      & 3   &               &      \\
\textbf{Covariance Type:} & HC3 &             &      \\
\bottomrule
\end{tabular}
\end{adjustbox}

\vspace{0.6em}

\begin{adjustbox}{max width=\textwidth}
\begin{tabular}{lcccccc}
 & \textbf{coef} & \textbf{std err} & \textbf{t} & \textbf{P$>|t|$} & \textbf{[0.025} & \textbf{0.975]} \\
\midrule
\textbf{Hydro\_Water\_Reservoir} & 0.0003 & 2.25e-05 & 13.683 & 0.000 & 0.000 & 0.000 \\
\textbf{Fossil\_Hard\_Coal}      & 0.0011 & 0.000    & 3.918  & 0.000 & 0.001 & 0.002 \\
\bottomrule
\end{tabular}
\end{adjustbox}

\vspace{0.6em}

\begin{adjustbox}{max width=\textwidth}
\begin{tabular}{lclc}
\textbf{Omnibus:}       & 164.096 & \textbf{Durbin--Watson:} & 1.665 \\
\textbf{Prob(Omnibus):} & 0.000   & \textbf{Jarque--Bera (JB):} & 1162.067 \\
\textbf{Skew:}          & -1.468  & \textbf{Prob(JB):} & 4.58e-253 \\
\textbf{Kurtosis:}      & 10.529  & \textbf{Cond. No.:} & 16.2 \\
\bottomrule
\end{tabular}
\end{adjustbox}

\footnotesize
\medskip

\emph{Notes:} [1] $R^2$ is computed without centering (uncentered) since the model does not contain a constant. [2] Standard errors are heteroscedasticity-robust (HC3).
\end{table*}

\begin{table*}[t]
\centering
\caption{OLS Regression Results for Third Cluster}
\label{tab:ols_results_Third_cluster}
\setlength\tabcolsep{3.5pt}
\renewcommand{\arraystretch}{1.15}
\small

\begin{adjustbox}{max width=\textwidth}
\begin{tabular}{lclc}
\toprule
\textbf{Dep. Variable:}    & y   & \textbf{$R^2$ (uncentered):}        & 0.386 \\
\textbf{Model:}            & OLS & \textbf{Adj.\ $R^2$ (uncentered):}  & 0.354 \\
\textbf{Method:}           & Least Squares & \textbf{F-statistic:}      & 10.59 \\
\textbf{Date:}             & Mon, 21 Jul 2025 & \textbf{Prob (F-statistic):} & 1.19e-05 \\
\textbf{Time:}             & 16:16:35 & \textbf{Log-Likelihood:}    & -298.50 \\
\textbf{No. Observations:} & 61  & \textbf{AIC:}                 & 603.0 \\
\textbf{Df Residuals:}     & 58  & \textbf{BIC:}                 & 609.3 \\
\textbf{Df Model:}         & 3   &                                &       \\
\textbf{Covariance Type:}  & HC3 &                                &       \\
\bottomrule
\end{tabular}
\end{adjustbox}

\vspace{0.6em}

\begin{adjustbox}{max width=\textwidth}
\begin{tabular}{lcccccc}
 & \textbf{coef} & \textbf{std err} & \textbf{t} & \textbf{P$>|t|$} & \textbf{[0.025} & \textbf{0.975]} \\
\midrule
\textbf{Solar}         & 0.0095  & 0.004 &  2.427 & 0.015 & 0.002     & 0.017 \\
\textbf{Wind\_Onshore} & -0.0003 & 0.000 & -2.451 & 0.014 & -0.000    & -5.13e-05 \\
\textbf{Fossil\_Hard\_Coal} & 0.0060  & 0.002 &  3.443 & 0.001 & 0.003     & 0.009 \\
\bottomrule
\end{tabular}
\end{adjustbox}

\vspace{0.6em}

\begin{adjustbox}{max width=\textwidth}
\begin{tabular}{lclc}
\textbf{Omnibus:}       & 2.554 & \textbf{Durbin--Watson:}  & 1.780 \\
\textbf{Prob(Omnibus):} & 0.279 & \textbf{Jarque--Bera (JB):} & 2.001 \\
\textbf{Skew:}          & -0.063 & \textbf{Prob(JB):}       & 0.368 \\
\textbf{Kurtosis:}      & 3.878 & \textbf{Cond. No.:}       & 36.3 \\
\bottomrule
\end{tabular}
\end{adjustbox}

\footnotesize
\medskip

\emph{Notes:} [1] $R^2$ is computed without centering (uncentered) since the model does not contain a constant. [2] Standard errors are heteroscedasticity-robust (HC3).
\end{table*}

\begin{observation} \label{observation2}
    Controlling for specific production categories (types), the gas price has no direct effect on the Nordic SP.
\end{observation}

\begin{observation} \label{observation3}
Fossil hard coal is significant in the second and third clusters. When cheap generation sources cannot meet the demand for electricity, the market price follows the price of the next, more expensive source, hard coal.
\end{observation}

Table~\ref{tab:Drivers} summarizes the significant features for the three price clusters. These features are placed in matrix $\textbf{X}$ as defined in Section~\ref{sec:section2}. The following observation highlights the importance of cluster-based feature extraction. 

\begin{observation} \label{observation4}
    Different features have the power to explain Nordic SP variations in different price clusters. Recognizing these features improves prediction results.
\end{observation}

\begin{table}[t!]
\centering
\caption{Significant Features of Clusters}
\label{tab:Drivers}
\begin{tabular}{lp{8cm}}  
\toprule
\textbf{Cluster} & \textbf{Significant Features} \\  
\midrule
Cluster 1 & Biomass, Nuclear, Hydro Water Reservoir, Wind Onshore, Fossil Gas, Other, Auction Capacity Import \\
Cluster 2 & Hydro Water Reservoir, Fossil Hard Coal \\
Cluster 3 & Solar, Wind Onshore, Fossil Hard Coal \\
\bottomrule
\end{tabular}
\end{table}

\subsection{Addressing imperfect multicollinearity using Principal Component Analysis (PCA)}
The correlation structure of the hourly SPs shown in Fig. \ref{fig:Corr} indicates that the correlation is very high between certain delivery periods. Such multicollinearity is known to be detrimental to calculating regression model coefficients, as the data set matrix would be close to singular. 

To mitigate multicollinearity in the Nordic SP dataset, we propose to use Principal Component Analysis (PCA) as follows. Given the matrix $\mathbf{X}$ of $9$ significant features, we form the matrix $
[\mathbf{P}|\mathbf{X}]$. After centering the matrix $
[\mathbf{P}|\mathbf{X}]$, we compute the $24 + 9$ principal components by solving the following optimization problem 

\begin{subequations} \label{pca:opt}
\begin{align} 
    \max_{\boldsymbol{W}}&\quad tr[\boldsymbol{W}^T ([\mathbf{P}|\mathbf{X}]^T[\mathbf{P}|\mathbf{X}]) \boldsymbol{W}] \\ 
    \text{s.t.} &\quad\boldsymbol{w}_i^T\boldsymbol{w}_i = 1, \label{pca:opt1} \\
    &\quad \boldsymbol{w}_i^T\boldsymbol{w}_j = 0,~\forall~i\neq j, \label{pca:opt2}
\end{align}
\end{subequations}

where  $\boldsymbol{W}$ is the matrix consisting of $24 + 9$ weight vectors $\boldsymbol{w}$ and the matrix $\mathbf{\mathbf{Z}} = [\mathbf{P}|\mathbf{X}]\boldsymbol{W} = [\boldsymbol{v}_1,\boldsymbol{v}_2,...,\boldsymbol{v}_{24+9}]$ consists of $24 + 9$ principal components. The orthogonality constraint (\ref{pca:opt2}) ensures that the principal components  $\boldsymbol{v}$ are orthogonal. The principal weight vectors $\boldsymbol{w}$ are eigenvectors of the matrix $[\mathbf{P}|\mathbf{X}]^T[\mathbf{P}|\mathbf{X}]$ and the variance (importance) of each principal component is represented by the corresponding eigenvalue. We denote the $l$ principal components corresponding to the $l$ largest eigenvalues selected from matrix $\mathbf{Z}$ by $\mathbf{Z}(l)$. We will discuss in Section~\ref{sec:optimization} how to choose the number of components $l$.

The proposed forecast-optimized feature-engineering approach in this section will be used in the upcoming Section \ref{sec:multi-forecast} to provide the input to our forecast models.

\section{Multi-forecast selection-shrinkage algorithm}
\label{sec:multi-forecast}
In this section, we first establish the theoretical framework outlining the conditions that models must satisfy to be combined. The ambiguity decomposition \citep{krogh1995ensembles} motivates our selection criterion, which seeks complementary forecasts rather than merely accurate stand-alone models. Then we derive the optimal weights for model combinations, and finally, we will explain the optimization process. Guided by the insights from this theoretical foundation, in Section \ref{sec:Numerical}, we construct candidate models and demonstrate how our algorithm improves overall performance. 

\subsection{Selection phase}
In the regression model $\textbf{p} = \textbf{f}(\textbf{D}) + \eta$ with input matrix $\textbf{D} = [\textbf{P}|\textbf{X}]$ and output variable $\textbf{p}$, function $\textbf{f} : \mathbb{R}^{L \times c} \to \mathbb{R}^{24}$ should be approximated considering $\eta$ as some additive noise. In function $\textbf{f}$ definition, $c = 24 +9$, is the number of columns of $\textbf{D}$, $L$ is the look-back parameter. 

Suppose a collection of $K$ forecast models $\{\textbf{f}_k\}_{k=1}^{K}$ of $\textbf{p}$ is available. A weighted average forecast can be represented as $\bar{\textbf{f}}(\textbf{D}) = \sum_k \omega_k \textbf{f}_k(\textbf{D})$ where $\omega_k$s are positive and sum to one. 

For an input $\textbf{D}$, we define the error of this averaged forecast as $\epsilon(\textbf{D})=(\textbf{p}(\textbf{D}) - \bar{\textbf{f}}(\textbf{D}))^2$ and the error of the $k^{th}$ forecast model as $\epsilon_k(\textbf{D}) = (\textbf{p}(\textbf{D}) - \textbf{f}_k(\textbf{D}))^2$, and its ambiguity as $a_k(\textbf{D}) = (\textbf{f}_k(\textbf{D}) - \bar{\textbf{f}}(\textbf{D}))^2$. The following theorem follows: 

\begin{theorem} \label{thm1}
The error of the averaged forecast model can be decomposed as $\epsilon(\textbf{D}) = \sum_k \omega_k \epsilon_k(\textbf{D}) - \sum_k \omega_k a_k(\textbf{D})$. 
\end{theorem}

\begin{proof}
We have (suppressing $\textbf{D}$): 
\begin{align} \label{eq1}
    \sum_k \omega_k \epsilon_k - \sum_k \omega_k a_k &= \sum_k \omega_k (\textbf{p} - \textbf{f}_k)^2 - \sum_k \omega_k (\textbf{f}_k - \bar{\textbf{f}})^2 \notag \\
    & \hspace{-0pt} = \sum_k \omega_k \textbf{p} \cdot \textbf{p} - 2\textbf{p} \cdot \sum_k \omega_k \textbf{f}_k + \sum_k \omega_k \textbf{f}_k \cdot \textbf{f}_k \notag \\
    & \hspace{-0pt} - \sum_k \omega_k \textbf{f}_k \cdot \textbf{f}_k + 2\bar{\textbf{f}} \cdot \sum_k \omega_k \textbf{f}_k - \sum_k \omega_k \bar{\textbf{f}} \cdot \bar{\textbf{f}} \notag \\
    &= \textbf{p} \cdot \textbf{p} - 2\textbf{p} \cdot \bar{\textbf{f}} + \bar{\textbf{f}} \cdot \bar{\textbf{f}}
\end{align}
\end{proof}
The term $\sum_k \omega_k \epsilon_k(\textbf{D})$ in Theorem \ref{thm1} is the weighted average of the individual forecast errors, and the term $\sum_k \omega_k a_k(\textbf{D})$ is the weighted average of the ambiguities. Theorem \ref{thm1} shows that the more the forecast models differ from their averages, the lower the error $\epsilon(\textbf{D})$ will be, provided the individual errors remain constant. Alternatively, the more variance a given model can introduce to the final model, the lower the final model's error. 

\begin{corollary}\label{crl1}
The generalization error can be decomposed as $\epsilon = \bar{\epsilon} - \bar{a}$, where $\bar{\epsilon}$ is the average generalization error of individual models and $\bar{a}$ is the average generalization of ambiguities. 
\end{corollary}

\begin{proof}
By integrating over the input distribution, we obtain the generalization error:
\begin{equation}
\begin{aligned}
\epsilon = \int \epsilon(\textbf{D}) \, dP &= \int \bar{\epsilon}(\textbf{D}) \, dP - \int \bar{a}(\textbf{D}) \, dP \\
&\hspace{-0em} = \sum_k \omega_k \int \epsilon_k(\textbf{D}) \, dP - \sum_k \omega_k \int a_k(\textbf{D}) \, dP \\
&= \sum_k \omega_k \epsilon_k - \sum_k \omega_k a_k = \bar{\epsilon} - \bar{a}
\end{aligned}
\end{equation}
\end{proof} 

\begin{remark} \label{rmk1}
     The error $\epsilon$ will decrease by introducing another model if introduced variance is more than introduced bias. 
\end{remark} 

Theorem \ref{thm1} not only explains how combining models can improve performance, but also offers guidance on building ensembles. Rather than focusing on the bias–variance trade-off within a single model, we should think of it at the ensemble level. Starting from a low-bias model, we can introduce new models that add diversity — or variance — to the ensemble. If managed carefully, this trade-off can reduce the overall error more effectively than optimizing individual models alone.

\subsection{Shrinkage phase}
Suppose we have selected $N \leq K$ forecast models that satisfy the property of Remark \ref{rmk1}. We want to find the optimal weights of these forecast models so that we get the lowest error variance for the weighted average forecast. The idea emerges from the fact that in the first stage with Remark \ref{rmk1} we increase the variance in the ensemble level, and with this second step we want to minimize the increased variance.  

We define the error vector $\mathbf{e} = (e_1, e_2, \ldots, e_N)^\top$ (difference between actual and forecast values) and the weight vector $\boldsymbol{\omega} = (\omega_1, \omega_2, \ldots, \omega_N)^\top$. We now present the following theorem, a well-known result from portfolio optimization \cite{markowitz1952portfolio}, which we reinterpret in our context. 
\begin{theorem} \label{thm2}
The minimum error variance (MEV) combination of forecast models is obtained by setting
\begin{equation} \label{opt}
 \quad \boldsymbol{\omega}_{\text{MEV}} = \frac{\mathbf{\Sigma}^{-1} \mathbf{1}}{\mathbf{1}^\top \mathbf{\Sigma}^{-1}\mathbf{1}}\text{,} \quad \sigma_{\text{MEV}}^2 = \frac{1}{\mathbf{1}^\top \mathbf{\Sigma}^{-1}\mathbf{1}} \text{,} 
\end{equation}
where $\mathbf{\Sigma}$ is the error covariance matrix $[cov(e_i,e_j)]$ and $\mathbf{1}$ is $ N \times 1$ unit vector.
\end{theorem}
We will use the following results to prove the theorem.
\begin{lemma}
Consider the constrained optimization problem 
\begin{align} \label{opt1}
\min_{\boldsymbol{\omega}} \{ \boldsymbol{\omega}^\top \mathbf{\Sigma} \boldsymbol{\omega}~|~\boldsymbol{\omega}^\top \mathbf{e} = e^*,~\boldsymbol{\omega}^\top \mathbf{1} = 1\}.
\end{align}

The solution is given by 
\begin{equation}
\begin{aligned} \label{opt2}
\boldsymbol{\omega}_{\text{opt}} &= \frac{B \mathbf{\Sigma}^{-1} \mathbf{1} - A \mathbf{\Sigma}^{-1} \mathbf{e} + e^* \left( C \mathbf{\Sigma}^{-1} \mathbf{e} - A \mathbf{\Sigma}^{-1} \mathbf{1} \right)}{D}, \\
\sigma_{\text{opt}}^2 &= \frac{B - 2e^* A + e^*{^2} C}{D},
\end{aligned}
\end{equation}
where $e^*$ is the target error for the average model, $A = \mathbf{1}^\top \mathbf{\Sigma}^{-1} \mathbf{e}$ is the weighted mean error, $B = \mathbf{e}^\top \mathbf{\Sigma}^{-1} \mathbf{e}$ is the $F$ ratio, $C = \mathbf{1}^\top \mathbf{\Sigma}^{-1}\mathbf{1}$ and $D = BC - A^2$.
\end{lemma}

\begin{proof}[Proof of Lemma 1]
The Lagrangian function of the problem (\ref{opt1}) is $\mathcal{L}(\boldsymbol{\omega}, \lambda_1, \lambda_2) = \boldsymbol{\omega}^\top \mathbf{\Sigma} \boldsymbol{\omega} + \lambda_1 \left( e^* - \boldsymbol{\omega}^\top \mathbf{e} \right) + \lambda_2 \left( 1 - \boldsymbol{\omega}^\top \mathbf{1} \right)$. The partial derivatives of $\mathcal{L}$ are 

\begin{align} 
&\frac{\partial \mathcal{L}}{\partial \boldsymbol{\omega}} = 2 \boldsymbol{\omega}^\top \mathbf{\Sigma} - \lambda_1 \mathbf{e^\top} - \lambda_2 \mathbf{1^\top} = 0, \label{eq:omega} \\ 
& \frac{\partial \mathcal{L}}{\partial \lambda_1} = e^* - \boldsymbol{\omega}^\top \mathbf{e} = 0, \label{eq:lambda1} \\ 
&\frac{\partial \mathcal{L}}{\partial \lambda_2} = 1 - \boldsymbol{\omega}^\top \mathbf{1} = 0. \label{eq:lambda2} 
\end{align}

We multiply equation \eqref{eq:omega} by $\mathbf{\Sigma}^{-1}$ and then by $\mathbf{e}$ from the right and using equations \eqref{eq:lambda1} and \eqref{eq:lambda2}, we obtain
\begin{equation} 
\begin{aligned} \label{first}
    2e^* = \lambda_1 \mathbf{e^\top} \mathbf{\Sigma}^{-1} \mathbf{e} + \lambda_2 \mathbf{1^\top} \mathbf{\Sigma}^{-1} \mathbf{e}
          = \lambda_1 B + \lambda_2 A.
\end{aligned}
\end{equation}

By multiplying equation \eqref{eq:omega} by $\mathbf{\Sigma}^{-1}$ and then by $\mathbf{1}$ from right and substituting from equations \eqref{eq:lambda1} and \eqref{eq:lambda2}, we obtain
\begin{equation} 
\begin{aligned} \label{second}
    2 = \lambda_1 \mathbf{e^\top} \mathbf{\Sigma}^{-1}\mathbf{1}  + \lambda_2 \mathbf{1^\top} \mathbf{\Sigma}^{-1} \mathbf{1} 
          = \lambda_1 A + \lambda_2 C.
\end{aligned}
\end{equation}

From equations \eqref{first} and \eqref{second}, we can express $\lambda_1$ and $\lambda_2$ as
\begin{equation}
    \lambda_1 = 2\frac{e^* C-A}{D}, \lambda_2 = 2\frac{B - e^* A}{D}.
\end{equation}
By substituting $ \lambda_1$ and $ \lambda_2$ in equation \eqref{eq:omega} we get the results.
\end{proof}
We are now ready to prove Theorem~\ref{thm2}.
\begin{proof}[Proof of Theorem 2]
Observe that $\sigma_{\text{opt}}^2$ is a convex function of $e^*$. Hence,  we can find the minimum $\sigma_{\text{opt}}^2$ by setting the derivative of equation \eqref{opt2} with respect to $e^*$ equal to zero,
\begin{equation}
    \frac{\partial \sigma_{\text{opt}}^2}{\partial e^*} = \frac{-2A + 2e^*C}{D} = 0.
\end{equation}
Solving for $e^*$ and then substituting in \eqref{opt2}, we obtain the result.
\end{proof}

\subsection{Optimization}
\label{sec:optimization}
Fig. \ref{fig:model_flowchart} illustrates the proposed multi-forecast selection-shrinkage algorithm. In the \textbf{Selection phase}, the forecasting models that satisfy the property in Remark \ref{rmk1} are selected from a collection of forecasting models. Then, in the \textbf{Shrinkage phase}, the optimal weights of the selected forecast models are computed using Theorem \ref{thm2} to obtain the weighted forecast model. If the selected model is a statistical model, we use  $[\textbf{P} |\textbf{X}]$ as its input, whereas for machine learning models, the input data is the reduced feature matrix $\textbf{Z}(l)$. To obtain the optimal number of components $l$, we use grid search, i.e., we compute the RMSE for the final model for $l\in\{1,\ldots,24+9\}$, and we choose the smallest $l$ that minimizes the RMSE. 

\begin{figure}[H]
\centering
\includegraphics[width=\linewidth,height=5cm,keepaspectratio]{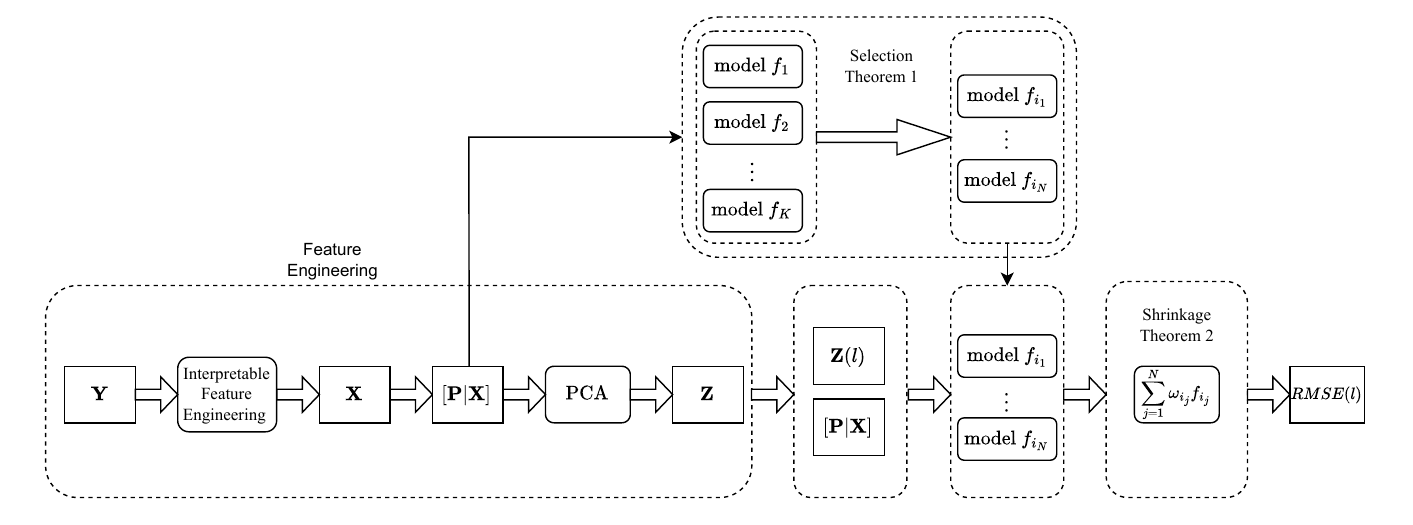}
\caption{Overview of the proposed multi-forecast PCA-selection-shrinkage algorithm.}
\label{fig:model_flowchart}
\end{figure}

\section{Numerical Experiments}
\label{sec:Numerical}
In this section, we demonstrate the effectiveness of our proposed algorithm on the Nordic SP dataset. We first specify and optimize candidate forecasting models for SP prediction. Using Theorem~\ref{thm1} and Remark~\ref{rmk1}, we then identify the subset of models that satisfy the bias--variance trading property at the ensemble level. Next, we show how the proposed shrinkage procedure, with and without PCA, reduces forecast error and improves model stability. We then compare the performance of our approach against two classes of benchmarks: (i) the state-of-the-art single model Temporal Fusion Transformer (TFT) and (ii) several widely used and practical forecast combination techniques. Finally, we perform a robustness check using synthetic data to confirm that the main results hold under resampling and that our conclusions are not sensitive to a particular sample period.

\subsection{Selection phase}
\subsubsection{The VARX model} \label{subsub1}

Our first candidate model is the Vector Autoregressive model with exogenous variables (VARX)~\cite{lutkepohl2005new}, selected for three main reasons. First, VARX models are widely used in electricity price forecasting and provide a well-established baseline~\cite{marcjasz2018selection,lehna2022forecasting}. Second, forecasting SP data at hourly frequency faces an endogeneity issue: all 24 hourly prices are published simultaneously around 1~pm on the day before delivery. Consequently, hourly lagged hourly prices cannot be used as predictors, and we must forecast the entire 24-dimensional price vector at daily frequency, precisely the setting where VARX provides a suitable multivariate framework. VARX($p,q$) is computationally efficient and captures cross-hour dependencies via its multivariate structure. Third, as noted in Remark~\ref{rmk1}, since statistical models such as VARX are theoretically high-bias but low-variance, this characteristic makes them valuable components of an ensemble because they complement the low-bias, high-variance neural network models.

The reduced-form VARX($p,q$) model is given by
\begin{align}
\zeta_t & = \mu + A_1\zeta_{t-1}+\dots+A_p\zeta_{t-p}+B_1X_{t-1}+\dots+B_qX_{t-q} + \epsilon_t \notag\\
  & = \mu + A(L)\zeta_t + B(L)X_t + \epsilon_t,
\end{align}
where $\zeta_{t-i}$ denotes the weekly-differenced system prices $(\Delta_7 \textbf{p}_t)^\top$, and $X_{t-i}$ denotes the weekly-differenced exogenous variables $(\Delta_7 \textbf{x}_{t-i})^\top$. We use the drivers identified in Section~\ref{sec:Section IV} (Table~\ref{tab:Drivers}) as explanatory variables and, based on the Bayesian Information Criterion (BIC), fit a VARX(10,1) model to the training data.

The first row of Table~\ref{tab:Reulst_RMSE} reports the RMSE for each of the 24 delivery periods. Errors are lowest for the first six periods, while hours 7–16 exhibit higher RMSE. The overall average RMSE is 15.368 with a standard deviation of 3.968. While VARX underperforms relative to the neural network models in isolation (reflecting its higher model bias), its inclusion is crucial for capturing the volatility of the price and improving ensemble performance.

\begin{table*}
\centering
\caption{Forecasting RMSE of the proposed C-PCA-SS algorithm, input models, alternative combinations, TFT, and ablations.}
\label{tab:Reulst_RMSE}
\resizebox{\textwidth}{!}{%
\begin{tabular}{c c c c c c c c c c c c c c c c c c c c c c c c}
\toprule
& 0-1 & 1-2 & 2-3 & 3-4 & 4-5 & 5-6 & 6-7 & 7-8 & 8-9 & 9-10 & 10-11 & 11-12\\
\hline
\midrule
VARX(10,1) & 6.62 & 8.44 & 9.32 & 10.22 & 10.79 & 11.33 & 13.6 & 20.17 & 21.13 & 18.75 & 18.4 & 18.21  \\
LSTM & 8.65 & 9.18 & 9.12 & 9.1 & 8.75 & 8.98 & 11.82 & 20.75 & 19.59 & 15.71 & 16.13 & 16.7 \\
LSTM-XGBoost & 7.6 & 8.02 & 8.47 & 8.78 & 8.68 & 9.0 & 11.52 & 17.75 & 15.25 & 14.98 & 16.48 & 16.51 \\
LSTM-CNN & 7.33 & 7.97 & 8.06 & 8.44 & 8.52 & 8.93 & 11.17 & 18.63 & 17.04 & 15.68 & 17.68 & 18.75 \\
\cellcolor{gray!30} \textbf{C-SS} & \cellcolor{gray!30} 6.66 & \cellcolor{gray!30} 7.4 &\cellcolor{gray!30} 7.95 &\cellcolor{gray!30} 8.48 &\cellcolor{gray!30} 8.5 &\cellcolor{gray!30} 8.93 &\cellcolor{gray!30} 10.97 &\cellcolor{gray!30} 17.45 &\cellcolor{gray!30} 15.94 &\cellcolor{gray!30} 14.94 &\cellcolor{gray!30} 15.97 &\cellcolor{gray!30} 15.79 \\
\cellcolor{gray!30} \textbf{C-PCA-SS} & \cellcolor{gray!30} 6.13 & \cellcolor{gray!30} 6.96 & \cellcolor{gray!30} 7.74 & \cellcolor{gray!30} 7.92 & \cellcolor{gray!30} 8.21 & \cellcolor{gray!30} 8.83 & \cellcolor{gray!30} 10.26 & \cellcolor{gray!30} 17.02 & \cellcolor{gray!30} 15.22 & \cellcolor{gray!30} 14.54 & \cellcolor{gray!30} 16.08 & \cellcolor{gray!30} 15.33 \\

\cellcolor{gray!30} \textbf{C-PCA/e-SS} & \cellcolor{gray!30} 7.36 & \cellcolor{gray!30} 8.3 & \cellcolor{gray!30} 8.5 & \cellcolor{gray!30} 9.0 & \cellcolor{gray!30} 9.27 & \cellcolor{gray!30} 9.63 & \cellcolor{gray!30} 11.85 & \cellcolor{gray!30} 19.21 & \cellcolor{gray!30} 18.0 & \cellcolor{gray!30} 16.6 & \cellcolor{gray!30} 16.42 & \cellcolor{gray!30} 16.49 \\
\cellcolor{gray!30} \textbf{PCA-SS} & \cellcolor{gray!30} 6.05 & \cellcolor{gray!30} 6.77 & \cellcolor{gray!30} 7.8 & \cellcolor{gray!30} 8.06 & \cellcolor{gray!30} 8.5 & \cellcolor{gray!30} 9.07 & \cellcolor{gray!30} 10.05 & \cellcolor{gray!30} 17.43 & \cellcolor{gray!30} 15.8 & \cellcolor{gray!30} 15 & \cellcolor{gray!30} 16.21 & \cellcolor{gray!30} 15.83 \\
\cellcolor{gray!30} \textbf{C-PCA-SS-h} & \cellcolor{gray!30} 4.99 & \cellcolor{gray!30} 6.48 & \cellcolor{gray!30} 7.32 & \cellcolor{gray!30} 7.53 & \cellcolor{gray!30} 8.09 & \cellcolor{gray!30} 8.7 & \cellcolor{gray!30} 9.99 & \cellcolor{gray!30} 16.88 & \cellcolor{gray!30} 14.8 & \cellcolor{gray!30} 14.31 & \cellcolor{gray!30} 15.97 & \cellcolor{gray!30} 15.16 \\
SimAV & 7.72 & 8.62 & 8.05 & 8.43 & 8.45 & 9.87 & 10.65 & 17.96 & 16.93 & 14.92 & 15.91 & 16.12 \\
CLS & 6.81 & 7.97 & 8.44 & 8.97 & 8.85 & 9.4 & 11.43 & 17.45 & 16.94 & 15.94 & 15.97 & 16.79 \\
IRMSE & 6.80 & 7.74 & 8.35 & 9.41 & 9.79 & 10.01 & 10.68 & 17.67 & 16.59 & 15.40 & 15.95 & 16.50 \\
NN-Comb & 6.79 & 7.89 & 8.32 & 8.72 & 8.87 & 9.39 & 10.88 & 17.05 & 15.98 & 14.22 & 15.30 & 15.48 \\
TFT & 5.73 & 7.41 & 8.08 & 8.10 & 8.24 & 8.53 & 11.13 & 17.47 & 16.98 & 17.23 & 17.80 & 18.10 \\
\hline
\midrule
& 12-13 & 13-14 & 14-15 & 15-16 & 16-17 & 17-18 & 18-19 & 19-20 & 20-21 & 21-22 & 22-23 & 23-0 & average \\
\hline
\midrule
VARX(10,1) & 17.56 & 17.97 & 18.44 & 18.85 & 18.7 & 16.66 & 15.52 & 16.32 & 16.69 & 15.35 & 14.65 & 15.15 & 15.37 \\
LSTM & 16.53 & 16.62 & 16.55 & 16.58 & 15.7 & 15.58 & 14.85 & 16.16 & 16.71
& 14.34 & 12.38 & 12.83 & 14.14\\
LSTM-XGBoost & 15.51 & 16.18 & 16.58 & 16.21 & 15.62 & 14.53 & 13.6 & 15.44 & 14.55 & 12.31 & 12.51 & 12.55 & 13.28\\
LSTM-CNN & 18.69 & 18.34 & 17.88 & 17.42 & 16.01 & 14.94 & 14.1 & 15.23 & 15.24 & 13.13 & 11.56 & 12.5 & 13.9\\
\cellcolor{gray!30} \textbf{C-SS} & \cellcolor{gray!30} 14.7 & \cellcolor{gray!30} 15.32 & \cellcolor{gray!30} 15.55 & \cellcolor{gray!30} 15.11 & \cellcolor{gray!30} 14.93 & \cellcolor{gray!30} 13.63 & \cellcolor{gray!30} 12.7 & \cellcolor{gray!30} 14.54 & \cellcolor{gray!30} 13.86 & \cellcolor{gray!30} 11.65 & \cellcolor{gray!30} 11.32 & \cellcolor{gray!30} 11.84 & \cellcolor{gray!30} 12.69 \\
\cellcolor{gray!30} \textbf{C-PCA-SS} & \cellcolor{gray!30} 14.76 & \cellcolor{gray!30} 15.27 & \cellcolor{gray!30} 15.39 & \cellcolor{gray!30} 14.99 & \cellcolor{gray!30} 14.57 & \cellcolor{gray!30} 13.21 & \cellcolor{gray!30} 11.88 & \cellcolor{gray!30} 13.1 & \cellcolor{gray!30} 13.55 & \cellcolor{gray!30} 11.45 & \cellcolor{gray!30} 10.9 & \cellcolor{gray!30} 11.31 & \cellcolor{gray!30} 12.28 \\

\cellcolor{gray!30} \textbf{C-PCA/e-SS} & \cellcolor{gray!30} 16.49 & \cellcolor{gray!30} 16.71 & \cellcolor{gray!30} 16.35 & \cellcolor{gray!30} 14.95 & \cellcolor{gray!30} 14.40 & \cellcolor{gray!30} 13.21 & \cellcolor{gray!30} 11.7 & \cellcolor{gray!30} 12.9 & \cellcolor{gray!30} 13.18 & \cellcolor{gray!30} 11.3 & \cellcolor{gray!30} 10.63 & \cellcolor{gray!30} 12.54 & \cellcolor{gray!30} 13.12\\
\cellcolor{gray!30} \textbf{PCA-SS} & \cellcolor{gray!30} 15.04 & \cellcolor{gray!30} 15.46 & \cellcolor{gray!30} 15.11 & \cellcolor{gray!30} 15.1 & \cellcolor{gray!30} 14.4 & \cellcolor{gray!30} 13.45 & \cellcolor{gray!30} 12.02 & \cellcolor{gray!30} 12.95 & \cellcolor{gray!30} 13.27 & \cellcolor{gray!30} 11.02 & \cellcolor{gray!30} 11 & \cellcolor{gray!30} 11.57 & \cellcolor{gray!30} 12.37\\
\cellcolor{gray!30} \textbf{C-PCA-SS-h} & \cellcolor{gray!30} 14.7 & \cellcolor{gray!30} 15.16 & \cellcolor{gray!30} 15.13 & \cellcolor{gray!30} 14.93 & \cellcolor{gray!30} 14.35 & \cellcolor{gray!30} 13.19 & \cellcolor{gray!30} 11.63 & \cellcolor{gray!30} 12.7 & \cellcolor{gray!30} 13.18 & \cellcolor{gray!30} 11.15 & \cellcolor{gray!30} 10.44 & \cellcolor{gray!30} 11.2 & \cellcolor{gray!30} 11.99 \\
SimAV & 16.53 & 16.77 & 15.76 & 15.85 & 15.15 & 14.50 & 13.97 & 14.58 & 15.40 & 13.17 & 11.95 & 12.83 & 13.33 \\
CLS & 15.1 & 15.32 & 15.8 & 15.95 & 15.2 & 13.85 & 13.71 & 15.54 & 14.3 & 12.55 & 11.73 & 11.84 & 13.16 \\
IRMSE & 15.76 & 15.95 & 16.07 & 17.21 & 15.67 & 14.72 & 14.10 & 14.43 & 14.08 & 11.88 & 10.92 & 11.77 & 13.22 \\
NN-Comb & 15.39 & 15.81 & 15.95 & 16.22 & 15.83 & 15.25 & 13.54 & 15.02 & 16.18 & 14.36 & 12.90 & 13.92 & 13.30 \\
TFT & 18.59 & 16.77 & 15.35 & 15.22 & 17.33 & 14.15 & 12.02 & 13.22 & 14.16 & 11.17 & 10.35 & 12.25 & 13.14 \\
\bottomrule
\end{tabular}
}
\end{table*}

\subsubsection{The LSTM and LSTM-XGboost models} 
\label{subsub2}
Following Remark \ref{rmk1}, the next forecast model is a low-bias and potentially high-variance model. We choose the Long-Short-Term Memory (LSTM) model, which has been widely used in the recent literature ~\cite{lehna2022forecasting,lago2021forecasting}. By optimizing the number of layers, the number of neurons, and the hyperparameters listed in Table~\ref{tab:LSTM_XGBoost}, we obtained an LSTM($1\times58$) model consisting of one layer with 58 perceptrons. To better capture spatial relationships between different hours, we feed our in-sample residuals of the LSTM to an XGBoost forecast model. In this way, the whole model better captures the concurrent relationships between hours. It reduces prediction errors in both in-sample and out-of-sample evaluations. The out-of-sample residual predictions by XGBoost have been combined with LSTM out-of-sample predictions to obtain final forecasts. Table~\ref{tab:LSTM_XGBoost} shows our hyperparameters of the LSTM-XGBoost forecast model. 

The second and third rows of Table \ref{tab:Reulst_RMSE} show the RMSE of the LSTM and the LSTM-XGBoost models for different delivery periods. Except for 4 hours with negligible difference, the LSTM-XGBoost performs better than LSTM, reducing both errors and error variance. The LSTM-XGBoost model also outperforms VARX(10,1), except for a negligible difference in the first delivery period. The average RMSE (standard deviation) for all delivery periods for LSTM is 14.138 (4.083) and for LSTM-XGBoost is 13.277 (3.712).

\begin{table}[H]
\centering
\begin{minipage}[t]{0.48\textwidth}
    \centering
    \caption{The LSTM-XGBoost forecast model hyperparameters}
    \begin{tabular}{@{}ll@{}}
        \toprule
        \textbf{Hyperparameter} & \textbf{Value} \\ \midrule
        \multicolumn{2}{l}{\textbf{LSTM Parameters}} \\
        activation & tanh \\
        recurrent\_activation & linear \\
        use\_bias & True \\
        activity\_regularizer & L1.L2 \\
        recurrent\_regularizer & L1.L2 \\
        return\_sequences & False \\
        stateful & True \\
        dropout & 0.2 \\
        optimizer & RMSprop \\
        loss & huber \\ 
        look\_back & 20 days \\ \midrule
        \multicolumn{2}{l}{\textbf{XGBoost Parameters}} \\
        objective & reg:squarederror \\
        n\_estimators & 100 \\
        learning\_rate & 0.1 \\
        max\_depth & 5 \\ \bottomrule
    \end{tabular}
    \vspace{0cm}
    \footnotesize Note: For unspecified parameters, default values from TensorFlow 2.16.2 are used.
    \label{tab:LSTM_XGBoost}
\end{minipage}%
\hfill
\begin{minipage}[t]{0.48\textwidth}
    \centering
    \caption{The CNN model structure and hyperparameters}
    \begin{tabular}{@{}p{5cm}@{}}
        \toprule
        \textbf{Layer Configuration} \\ \midrule
        Conv1D(filters=64, kernel\_size=3, activation='tanh') \\
        MaxPooling1D(pool\_size=2) \\
        Dropout(rate=0.2) \\
        Conv1D(filters=128, kernel\_size=3, activation='tanh') \\
        MaxPooling1D(pool\_size=2) \\
        Dropout(rate=0.2) \\
        Flatten() \\
        Dense(units=64, activation='tanh') \\
        Dropout(rate=0.5) \\
        Dense(units=n\_hours, activation='linear') \\ \bottomrule
    \end{tabular}
    \vspace{0cm}
    \footnotesize Note: For unspecified parameters, default values from TensorFlow 2.16.2 are used.
    \label{tab:CNN_Model}
\end{minipage}
\end{table}

\subsubsection{The LSTM-CNN model} \label{subsub3}
As the fourth forecast model, we use a Convolutional Neural Network (CNN) and develop an LSTM-CNN model, following~\cite{lehna2022forecasting}, in the same way as the LSTM-XGBoost forecast model. Table~\ref{tab:CNN_Model} shows the CNN model structure and the hyperparameter values used in the LSTM-CNN model. The fourth row of Table \ref{tab:Reulst_RMSE} shows the performance of the LSTM-CNN model for different delivery periods. The LSTM-CNN model performs better in some delivery periods, but its performance is worse than that of the LSTM-XGBoost on average, and it has a higher variance than the LSTM-XGBoost model. The average RMSE over all delivery periods for the LSTM-CNN is 13.9, with a standard deviation of 3.982. 

\subsubsection{The ensemble level bias-variance property of Remark \ref{rmk1}}
In order to check the bias-variance property for the above forecast models, we minimize the in-sample RMSE of the combined forecast models with respect to model weights. For our data set, VARX and LSTM-XGBoost are the two forecast models that satisfy the property of Remark \ref{rmk1}, resulting in non-zero coefficients. In Fig. \ref{fig:model_flowchart}, this phase is represented as the selection phase, taking $[\mathbf{P}|\mathbf{X}]$ as input from the feature engineering phase and producing the selected models.

\subsection{Shrinkage phase}
 As is shown in Fig. \ref{fig:model_flowchart}, we apply PCA to the data set $[\mathbf{P}|\mathbf{X}]$ and obtain $\textbf{Z}$, then we choose the first $l$ columns of this matrix and form the matrix $\textbf{Z}(l)$ for each $l$. By using Theorem~\ref{thm2} and jointly optimizing the number of components, $l$, we obtain the joint variables $(l^*, \omega_{\text{LSTM-XGBoost}}) = (22, 0.7515)$. An important fact is that shrinkage coefficients are robust to small variations. The Clustering-PCA-Selection-Shrinkage (\textbf{C-PCA-SS}) row in Table~\ref{tab:Reulst_RMSE} shows the result in this case. As we can see, using the (\textbf{C-PCA-SS}) method reduces RMSE in \textbf{all} delivery periods compared to the input models. The proposed multi-forecast selection-shrinkage algorithm, augmented with PCA, improves the forecast accuracy, achieving an average RMSE of 12.28 with a standard deviation of 3.25. 

We also repeated the optimization of joint variables based on predicting a single hour (rather than the full vector). While this approach outperformed the vector-based one across all delivery periods, the performance gains were offset by the significantly higher computational cost. The Clustering-PCA-Selection-Shrinkage-hourly (\textbf{C-PCA-SS-h}) row in Table~\ref{tab:Reulst_RMSE} shows the obtained results for this studied case. 

\subsection{Ablation Study}
We present the performance of alternative models in this subsection to show the efficacy and importance of each sub-process we have introduced in (\textbf{C-PCA-SS}). The first alternative is to exclude PCA from the process. We used Theorem~\ref{thm2} to calculate the optimal weights for the two forecast models. In the case of not applying PCA to the data set, we obtain $\omega_{\text{LSTM-XGBoost}}=0.7529$ and $\omega_{\text{VARX}}=0.2471$.  
Notably, the results of the LSTM-XGBoost and VARX shrinkage forecast models are robust to variations in these weights. The fifth row of Table \ref{tab:Reulst_RMSE} shows the performance of the Clustering, Selection, and Shrinkage (\textbf{C-SS}) forecast model for the 24-delivery periods. The shrinkage forecast model results in a noticeable reduction in RMSE for all but two delivery periods relative to the input models, while also reducing both the mean and variance. The average RMSE (standard deviation) for all delivery periods is 12.689 (3.647).

To compare the performance of the model with respect to selecting the number of components, we repeat the process with the number of components obtained by the elbow method. In this case, the number of components retained is 10. The Clustering-PCA(elbow)-Selection-Shrinkage (\textbf{C-PCA/e-SS}) row in Table~\ref{tab:Reulst_RMSE} shows the result in this case. As shown in the table, integrating PCA with the downstream task of RMSE minimization provides a great improvement against the widely used elbow method for determining the optimal number of components. 

Finally, to show the effectiveness of clustering-based feature selection, we use only significant features obtained without partitioning the data into 3 clusters. The PCA-Selection-Shrinkage (\textbf{PCA-SS}) row in Table~\ref{tab:Reulst_RMSE} shows the result for this model, and as is evident, capturing short-term price deviations by the relevant features enhances model performance. 

To avoid overfitting, we regularized the forecast models, specifically the machine learning models. Doing so, however, increases bias. Nonetheless, we compensate for the bias using the shrinkage method at the ensemble level. For example, while the VARX model does not perform well relative to machine learning models, through the proposed shrinkage method, it does help LSTM-XGBoost perform better. 

Having established the effectiveness of the proposed selection-shrinkage algorithm and its PCA-augmented variants, we now compare these results against other advanced benchmarks. First, we evaluate the Temporal Fusion Transformer (TFT) as a recent state-of-the-art deep learning model to assess how our proposed approach performs relative to this recent deep learning model. Second, we examine widely used forecast combination methods, including (1) simple averaging, (2) constrained least squares, (3) inverse RMSE weighting, and (4) neural network–based combinations. These additional benchmarks provide a comprehensive assessment of the strengths of our proposed approach.

\subsection{Further benchmark assessment}
\subsubsection{Temporal Fusion Transformer (TFT) model}
The Temporal Fusion Transformer (TFT) is a state-of-the-art deep learning architecture for multi-horizon time series forecasting, first introduced by \cite{lim2021temporal}, and it has recently gained attention in electricity price forecasting~\cite{jiang2024probabilistic, khan2024transformer, putz2024feasibility}. Prior studies provide mixed evidence regarding its performance: \cite{ganesh2023forecasting} finds that TFT does not outperform simpler models such as LSTMs or fully connected networks, while \cite{deng2024seasonality} reports that TFT delivers results comparable to, but not consistently better than, BiLSTM, LSTM, and Temporal Convolutional Network (TCN) models.  

Motivated by these findings, we implemented a TFT model based on its original design and optimized the hyperparameters listed in Table~\ref{tab:tft_params}.\footnote{Because the hyperparameter space is high-dimensional, finding a globally optimal configuration remains challenging.} The resulting model achieved an average RMSE of 13.141 with a standard deviation of 3.926 across all delivery periods (Table~\ref{tab:Reulst_RMSE}). Although each of our base models is comparatively simple and underperforms TFT in isolation, our proposed selection--shrinkage procedure combines them into an ensemble that achieves a consistently lower RMSE than TFT. This result underscores the value of our systematic approach, which leverages complementary model strengths to outperform even state-of-the-art deep learning baselines.

\begin{table}[H]
    \centering
    \caption{Key Hyperparameters for TFT Model}
    \label{tab:tft_params}
    \footnotesize
    \begin{tabular}{@{}ll@{}}
        \toprule
        \textbf{Hyperparameter} & \textbf{Value} \\ \midrule
        \multicolumn{2}{l}{\textbf{Model Architecture}} \\
        \texttt{hidden\_size} & 27 \\
        \texttt{lstm\_layers} & 1 \\
        \texttt{num\_attention\_heads} & 2 \\
        \texttt{full\_attention} & True \\
        \texttt{feed\_forward} & GRN \\
        \texttt{dropout} & 0.2 \\
        \texttt{regularization} & L1/L2\\
        \texttt{hidden\_continuous\_size} & 27 \\
        \texttt{categorical\_embedding\_sizes} & None \\
        \texttt{norm\_type} & LN \\
        \texttt{input\_chunk\_length} & 20 \\
        \texttt{output\_chunk\_length} & 24 \\
        \multicolumn{2}{l}{\textbf{Training Parameters}} \\
        \texttt{batch\_size} & 30 \\
        \texttt{n\_epochs} & 100 \\
        \texttt{loss\_fn} & Huber \\
        \texttt{optimizer} & Adam \\
        \texttt{random\_state} & 65 \\
        \bottomrule
    \end{tabular}

    \vspace{1mm}
    \footnotesize
    \textit{Note: The model also uses EarlyStopping and ReduceLROnPlateau callbacks. GRN: Gated Residual Network. LN: Layer Normalization.}
\end{table}

\subsubsection{Alternative combination methods}
To benchmark our selection-shrinkage algorithm against established practice, we evaluate four alternative forecast combination methods. To ensure a fair and controlled comparison, we apply our feature selection procedure before evaluating all combination methods. The methods are: (i) simple averaging; (ii) Constrained Least Squares (CLS); (iii) inverse  RMSE (IRMSE) weighting; and (iv) a neural network combiner. We briefly describe each method below and report comparative results in Table~\ref{tab:Reulst_RMSE}.

\subsubsection{Simple Averaging}
In this case the the forecasted day-ahead price will be the simple average of the forecasted values of each of the single models:

\begin{equation}
    \textbf{p}^{^{SA}}_d = \frac{\textbf{p}^{VARX}_d + \textbf{p}^{LSTM}_d + \textbf{p}^{LSTM-XGBoost}_d + \textbf{p}^{LSTM-CNN}_d}{4}, 
\end{equation}

where, for example, $\textbf{p}^{VARX}_d$ is the hourly day ahead price vector for day $d$ predicted by VARX. The SimAV row of the Table~\ref{tab:Reulst_RMSE} shows the results of this method.

\subsubsection{Constrained Least Squares (CLS)}
The Constrained Least Squares (CLS) combination method addresses two well–known limitations of Ordinary Least Squares (OLS) in forecast pooling: (i) strong correlations among base forecasts inflate the variance of OLS weights, and (ii) unconstrained OLS can assign negative weights that are hard to interpret in this context \citep{theil1961pure,lawson1974solving}. The CLS method remedies both, thus stabilizing estimation and improving interpretability.

Let $\mathcal{M}$ denote the set of base models and $\mathbf{p}^{(m)}_d$ their forecast vectors for day $d$. The CLS combiner is
\begin{align}
\mathbf{p}^{\text{CLS}}_d &= \sum_{m\in\mathcal{M}} \omega_m\, \mathbf{p}^{(m)}_d + \mathbf{e}_d, 
& \text{s.t. }\; \mathbf{1}^\top \boldsymbol{\omega}=1,\;\boldsymbol{\omega}\ge 0, \label{eq:cls}
\end{align}
where, in our application $m\in\{\text{VARX},\ \text{LSTM},\ \allowbreak \text{LSTM--XGBoost},\ \allowbreak \text{LSTM--CNN}\}$.

Weights $\boldsymbol{\omega}$ are estimated from in-sample forecasts using constrained least squares and then applied to out–of–sample forecasts. The results are reported in the \textit{CLS} row of Table~\ref{tab:Reulst_RMSE}.

\subsubsection{IRMSE}  

One of the most widely used and empirically successful forecast combination schemes is the inverse RMSE weighting method \citep{wang2023forecast, diebold1987structural}. The reasoning is straightforward: better-performing models in the past should be given larger influence in the combined forecast. It is implemented by giving weights proportional to the inverse of each model's RMSE, with biases towards lower prediction error models. In our context, for $i$ corresponding to one of 
\textit{VARX}, \textit{LSTM}, \textit{LSTM-XGBoost}, or \textit{LSTM-CNN}, 
the weight of model $i$ is given by:


\begin{equation}
  \omega_i = \frac{\text{a-RMSE}_i^{-1}}{\sum_j \text{a-RMSE}_j^{-1}}.
\end{equation}

where a-RMSE means average over 24 delivery periods. 

The weights $\omega_i$ can be calculated both dynamically (online) and statically. We prefer the second one, and calculate the weights based on the in-sample data, and use them for the out-of-sample evaluation. The IRMSE row of Table~\ref{tab:Reulst_RMSE} shows the results of this method.

\subsubsection{Neural network combining method}  

A more sophisticated combination approach regresses in-sample targets on base models' predictions using a neural network \cite{zhang2003time}. Unlike linear averaging, neural networks can learn nonlinear interactions among forecasts and the target \cite{wang2023forecast}. The trade-offs are reduced interpretability and a higher risk of overfitting, necessitating careful regularization and validation \cite{makridakis2020m4} and, unlike constrained linear regression, they do not enforce non-negativity of implicit weights while introducing substantial parameter and tuning burdens \cite{goodfellow2016deep}. In our implementation, we use an LSTM combining machine with the architecture as in Section \ref{subsub2} and the hyperparameters reported in Table~\ref{tab:LSTM_XGBoost}; results appear as “NN-Comb” in Table~\ref{tab:Reulst_RMSE}.

Table~\ref{tab:Reulst_RMSE} indicates that, among the four benchmark combination schemes, the CLS method delivers the best performance. The Neural Network combiner (NN-Comb), despite its added complexity, does not improve upon CLS or IRMSE weighting and only marginally outperforms simple averaging. Crucially, none of these alternatives performs better than our proposed selection–shrinkage method.

\subsection{Robustness check of the proposed forecast model}
\begin{table}
\centering
\caption{RMSE Results of Forecast Models on Synthetic Test Data.}
\label{tab:Reulst_Artificial_RMSE}
\resizebox{\textwidth}{!}{%
\begin{tabular}{c c c c c c c c c c c c c c c c c c c c c c c c}
\toprule
& 0-1 & 1-2 & 2-3 & 3-4 & 4-5 & 5-6 & 6-7 & 7-8 & 8-9 & 9-10 & 10-11 & 11-12 \\
\hline
\midrule
VARX(10,1) & 10.48 & 10.26 & 9.93 & 9.91 & 10.06 & 10.25 & 10.55 & 11.36 & 12 & 11.86 & 11.67 & 11.62  \\
LSTM-XGBoost & 9.74 & 9.66 & 9.4 & 9.16 & 9.47 & 9.93 & 10.46 & 11.86 & 11.89 & 11.39 & 11.38 & 10.74 \\
\textbf{C-PCA-SS} & 9.26 & 9.0 & 9.05 & 8.91 & 9.1 & 9.72 & 10.17 & 11.05 & 11.12 & 10.8 & 10.82 & 10.02 \\
\hline
\midrule
& 12-13 & 13-14 & 14-15 & 15-16 & 16-17 & 17-18 & 18-19 & 19-20 & 20-21 & 21-22 & 22-23 & 23-0 \\
\hline
\midrule
VARX(10,1)& 11.23 & 10.73 & 10.72 & 10.64 & 10.8 & 11  & 11.31 & 11.9 & 12.26 & 12.24 & 11.87 & 11.44 \\
LSTM-XGBoost & 10.43 & 10.33 & 10.16 & 9.98 & 10.21 & 10.96 & 11.57 & 11.52 & 11.2 & 11 & 10.79 & 10.12 \\
\textbf{C-PCA-SS} & 9.8 & 9.73 & 9.6 & 9.3 & 9.15 & 10.11 & 10.3 & 10.83 & 10.1 & 10 & 9.92 & 9.1 \\
\bottomrule
\vspace{-0.5cm}
\end{tabular}
}
\end{table}
We created synthetic data for the period April and May 2024 by sampling from the historical data, while preserving the statistical properties of the original dataset. We did so by choosing a daily SP vector $\bf{p}_d$ and a feature vector $\bf{x}_d$ from the same week of the year and day of the week from one of the previous years, at random. We generate $50$ synthetic data sets in this way.

Table~\ref{tab:Reulst_Artificial_RMSE} shows the average of the results over the $50$ synthetic datasets. The results show that the general findings on the usefulness of our proposed PCA-augmented selection-shrinkage approach remain unchanged.

\section{Conclusion} \label{sec:conc}
This paper develops a systematic framework for forecasting the Nordic System Price (SP) that combines interpretable driver identification with a principled ensemble methodology. We first statistically characterize the SP, documenting seasonal structure, cross-delivery-period dependence, and significant drivers, and then introduce a forecast-optimized feature-engineering pipeline that integrates $k$-means clustering, the MSTD method, and SARIMA. To address imperfect multicollinearity, we apply PCA and select the number of components by minimizing the downstream RMSE rather than relying on heuristic rules. 

Building on these inputs, we propose a PCA-augmented selection–shrinkage framework that combines complementary forecasting models using theoretically grounded, variance-minimizing weights. On Nordic data (2015–2024), the approach improves accuracy across all 24 delivery periods and consistently outperforms individual models, alternative combination schemes, and the state-of-the-art Temporal Fusion Transformer (TFT). In our main specification (\textbf{C-PCA-SS}), the ensemble achieves an average RMSE of 12.28 and standard deviation of 3.25, with robustness checks on 50 synthetic datasets confirming that these gains are stable. Importantly, these improvements are obtained at a computational cost comparable to the studied baseline models. 

For practitioners and policy makers, the proposed framework offers interpretable insights into price drivers and a more accurate day-ahead forecast of SP, supporting risk management and the design of hedge instruments such as EPADs. Although demonstrated on the Nordic SP, the methodology is general and can be applied to other interconnected electricity markets where cross-period dependencies are significant.

Two promising extensions remain for future research. First, the in-sample covariance matrix used in Theorem~\ref{thm2} could be replaced with structured or shrinkage covariance estimators tailored to the characteristics of the base models. 
Second, integrating spike forecasting modules into the proposed approach could further reduce RMSE and strengthen its performance.

\section*{Acknowledgments}
This research was supported by VINNOVA under project number 2022-01-425. 
All data used in this study are available through the Nord Pool data portal, accessible at \url{https://www.nordpoolgroup.com/}.

\begin{appendices} 
\section{The Canova-Hansen (CH) test} \label{app:1}
The CH test is conducted by means of the following regression model:
\begin{equation}
\label{equ:CH-regression}
(1 - B^d)y_t = \mu + f'_t \gamma + e_t,
\end{equation}
where $d$ is the order of integration of the long run of the series and $f'_t$ is a term that fits the deterministic seasonal component by means of trigonometric functions. In the case of the weekly seasonality, it is defined as:
\begin{equation}
f'_t = \left( \cos(\omega t), \sin(\omega t), \cos(2\omega t), \sin(2\omega t), \cos(3\omega t), \sin(3\omega t) \right)
\end{equation}
Under the alternative hypothesis of seasonal non-stationarity, the coefficients of the seasonal regressors are assumed to be time-variant. In a special case, they follow a random walk process:
\begin{equation}
\label{equ:Random-walk}
 \gamma_t = \gamma_{t-1} + u_t,
\end{equation}
where $u_t$ is an i.i.d. process independent of the $e_t$ with zero mean and a constant variance. If the null hypothesis of the stability in the seasonal coefficient is to be satisfied, then the covariance matrix of the process should be equal to zero. 

Canova and Hansen introduce the following statistics: a high value indicates rejection of the null hypothesis at frequencies specified by A.

\begin{equation}
L = \frac{1}{n^2} \sum_{i=1}^{n} \hat{F}'_i A (A'\hat{\Omega}A)^{-1} A' \hat{F}_i = 
\frac{1}{n^2} \operatorname{tr}\!\!\left[(A'\hat{\Omega}A)^{-1} A' \!\left(\sum_{i=1}^{n} \hat{F}_i \hat{F}'_i\right)\! A \right].
\end{equation}

where $\hat{F}_i =\sum_{t=1}^{i}f_t \hat{et_t}$ and $\hat{e}_t$ s are OLS residuals from equation (\ref{equ:CH-regression}) and A is $(s-1) \times a$ matrix that selects the $a$ elements of $\gamma_i$ that is to be tested for nonstationarity. If we want to test the stationarity alternative simultaneously at all seasonal frequencies, then $a = 6$ and $A =I_6$. When we want to test the stability hypothesis at a specific frequency $k\omega$ where $k = 1,2,3$, then $a=2$ and the A matrix would be either $A = (I_2,\mathbf{0},\mathbf{0})$ for $k = 1$ and $A = (\mathbf{0},I_2,\mathbf{0})$ for $k = 2$ and $A = (\mathbf{0},\mathbf{0},I_2)$ for $k = 3$, where $\mathbf{0}$ is $2\times2$ null matrix and the equation (\ref{equ:Random-walk}) should be modified as:

\begin{equation}
 A'\gamma_t = A'\gamma_{t-1} + u_t.
\end{equation}

 Finally, $\hat{\Omega}$ is defined in equation (\ref{equ:omega}), which is a semi-parametric heteroskedasticity and autocorrelation-consistent covariance estimator of the long-run variance of the process. 
 
 \begin{equation}
 \label{equ:omega}
 \hat{\Omega} = \sum_{k=-m}^{m}W(k/m) \frac{1}{n} \sum_{i=1}^{n}f_{i+k} \hat{e}_{i+k}f'_t \hat{e}_i,
\end{equation}

where W(.) is any kernel function like Bartlett or quadratic spectral that produces a positive semi-definite covariance matrix estimation \cite{canova1995seasonal}. 

\end{appendices}

\raggedright
\bibliography{references}

\end{document}